\documentclass[lettersize,journal]{IEEEtran}
\usepackage{stfloats}
\usepackage{amsmath,amsfonts}
\usepackage{algorithmic}
\usepackage{algorithm}
\usepackage{array}

\usepackage[caption=false,font=footnotesize,labelfont=rm,textfont=rm]{subfig}

\usepackage{textcomp}
\usepackage{stfloats}
\usepackage{url}

\usepackage[bookmarks=true, colorlinks=true, breaklinks=true]{hyperref}
\usepackage{empheq}

\usepackage{amsmath}
\usepackage{amsthm}

\newtheorem{prop}{Proposition}

\usepackage{color}
\usepackage{verbatim}
\usepackage{graphicx}
\usepackage{cite}
\usepackage{amssymb}
\graphicspath{{images/}}
\begin{document}
\title{Prior-Aided Iterative Channel Reconstruction with Optimized Frame Structure for DSE Mitigation in CP-OTFS-Based LEO Satellite Systems}
\author{Yiyan~Cheng,~Tiejun~Lv,~\IEEEmembership{Senior Member,~IEEE},~Yashuai~Cao,~Xuehan Wang,~and Mugen~Peng,~\IEEEmembership{Fellow,~IEEE}
\thanks{Manuscript received 01 July 2025; revised 07 February 2026; accepted 02 August 2026. This paper was supported in part by the National Natural Science Foundation of China under No. 62271068. (\emph{corresponding author: Tiejun Lv}.)}
\thanks{Y. Cheng, T. Lv and M. Peng are with the School of Information and Communication Engineering, Beijing University of Posts and Telecommunications (BUPT), Beijing 100876, China (e-mail: \{yiyancheng, lvtiejun, pmg\}@bupt.edu.cn).}
\thanks{Y. Cao is with the School of Artificial Intelligence, University of Science and Technology Beijing, Beijing 100083, China, and also with Beijing Key Laboratory of Intelligent and Communication Integration and Hebei Key Laboratory of Space-Air-Ground Intelligent Communication, University of Science and Technology Beijing, Beijing 100083, China (e-mail: caoys@ustb.edu.cn).}
\thanks{X. Wang is with the Department of Electronic Engineering, Tsinghua University, Beijing 100084, China, and also with the China Mobile Research Institute, Beijing 100053, China (e-mail: wangxuehan@chinamobile.com).}
}
\maketitle

\begin{abstract}
Orthogonal time frequency space (OTFS) modulation has emerged as a promising solution to mitigate the severe Doppler shift in low Earth orbit (LEO) satellite communications. However, the frequency-dependent Doppler shift induced by the high mobility of LEO satellites leads to the Doppler squint effect (DSE). 
This effect compromises the channel sparsity in the delay-Doppler (DD) domain, rendering existing channel estimation methods ineffective. To overcome this challenge, this paper proposes a DSE-resilient transmission scheme for cyclic prefix OTFS (CP-OTFS)-based LEO satellite systems.
Specifically, we analyze the input-output relationship of the CP-OTFS-based LEO satellite communication system and derive a DSE-aware representation of the satellite-terrestrial channel in the DD domain. 
To efficiently capture DSE-aware channel characteristics, we propose a novel OTFS frame structure that allows the energy distribution of the received signal to serve as prior information for channel estimation.
Meanwhile, this frame structure strategically allocates pilot symbols to achieve uniform energy distribution and reduce the peak-to-average power ratio (PAPR), while imposing a time-domain waveform continuity constraint to suppress out-of-band emission (OOBE) caused by rectangular pulses.
Based on the frame structure, we propose a prior-aided iterative channel reconstruction (PAICR) algorithm to mitigate the severe power leakage induced by DSE. The proposed algorithm iteratively extracts and removes dominant channel components using Doppler-domain received signal energy observations, with a convergence criterion ensuring reliable termination. Furthermore, a Cramér-Rao lower bound is derived to provide a theoretical benchmark for evaluating the algorithm’s performance.
Simulation results demonstrate that the proposed approach reduces power spectral density by at least 12 dB/Hz and PAPR by at least 3 dB over the compared benchmark schemes, and the proposed PAICR algorithm achieves noticeable the normalized mean square error improvements compared with existing methods.
\end{abstract}


\section{Introduction}
\IEEEPARstart{L}{ow} Earth orbit (LEO) satellite communication has emerged as a pivotal enabler for next-generation wireless ecosystems, thanks to its inherent merits of ubiquitous connectivity and latency propagation~\cite{LEO_advantage1, LEO_advantage2, LEO_advantage3}. However, the high orbital velocity (typically operating at velocities exceeding 7.5~\text{km/s}~\cite{LEO_Speed}) of LEO satellites gives rise to severe Doppler spread, which fundamentally compromises the orthogonality among subcarriers in orthogonal frequency division multiplexing (OFDM) systems. This phenomenon causes substantial inter-carrier interference (ICI) and hence limits the applicability of OFDM in LEO satellite systems~\cite{Perf_degradation}. To address this issue, orthogonal time frequency space (OTFS) modulation has emerged as a promising solution. By transforming the time-varying channel into the delay-Doppler (DD) domain, OTFS effectively mitigates the high Doppler effect and improves communication performance in such dynamic scenarios~\cite{9508932}. 

Nevertheless, the majority of existing studies on OTFS-based LEO satellite systems are based on the traditional sparse multipath channel model in the DD domain.
In~\cite{8836636}, an iterative channel estimation and data detection scheme was proposed based on the approximate message passing (AMP) algorithm.
Furthermore, a vector AMP (VAMP) algorithm for handling high-dimensional sparse reconstruction problems was proposed in~\cite{8713501}.
However, the presence of fractional Doppler spreads the channel energy along the Doppler dimension, rendering the channel less strictly sparse.
In response to this issue, some recent studies have investigated the impact of fractional Doppler in OTFS systems, and corresponding estimation methods have been developed to recover fractional Doppler channels. In~\cite{9738478}, an off-grid channel estimation scheme was proposed based on the sparse Bayesian learning (SBL) framework. To alleviate the associated computational burden, a fast Bayesian compressive sensing (FBCS) algorithm was further proposed in~\cite{10506450}.

In reality, as reported in~\cite{9398858}, significant Doppler differences among subcarriers exist since the Doppler shift induced by high mobility is correlated in the frequency domain. This critical phenomenon, referred to as the Doppler squint effect (DSE), has been generally neglected in the OTFS-based LEO satellite systems. Furthermore, the DD-domain channel model utilized in~\cite{8836636, 8713501, 9738478, 10506450} no longer holds true due to the power leakage caused by DSE. The resultant DSE-unaware approaches may degrade satellite-terrestrial channel estimation accuracy.

To handle the DSE, limited research efforts have been dedicated to OTFS-based channel estimation~\cite{10103827, 10475894, 10870146, 10970022, 10680151, 8701706, 7925924, 10464897}. For instance, the authors in~\cite{10103827} analyzed the impact of DSE on OTFS channel coefficients and developed an orthogonal matching pursuit (OMP)-based channel estimation approach. However, this method requires utilizing the entire OTFS frame for channel estimation, which leads to extremely low transmission efficiency. In~\cite{10475894}, a pilot scheme for channel estimation and data detection was proposed to mitigate interference between data and pilot symbols caused by the DD-domain expansion due to DSE. 
A deep learning-based DSE-resilient receiver for underwater OTFS communications was developed in~\cite{10870146}. Additionally, the authors in~\cite{10970022} proposed a Newton-based channel estimation algorithm that accurately estimates fractional delays and Doppler shifts by leveraging DSE characteristics. Nevertheless, the existing studies in~\cite{10475894, 10870146, 10970022} relied on rectangular pulses, which introduces considerable out-of-band emission (OOBE) and consequently degrade the spectral efficiency of OTFS systems.
In~\cite{10680151}, the authors proposed mitigating DSE interference by applying an interference phase shift matrix. Although this approach enhances OTFS performance as the number of DD plane bins increases, it incurs a linearly increasing peak-to-average power ratio (PAPR) with respect to the number of Doppler bins in the DD plane, as reported in~\cite{8701706}.
Moreover, the authors in~\cite{7925924} proposed to implement OTFS system by incorporating pre-processing and post-processing modules into the OFDM system. In~\cite{10464897}, the authors derived the input-output relation for cyclic prefix (CP)-OFDM-based OTFS systems in the discrete time domain under DSE conditions. However, a detailed analysis of the corresponding input-output relation in the DD domain was not explicitly explored.

In OTFS-based LEO satellite communications, DSE-aware transmission studies have yet to address the challenges of high OOBE and elevated PAPR. Furthermore, channel estimation techniques that effectively account for DSE characteristics remain insufficient.
Motivated by this, we first investigate the impact of DSE on the satellite-terrestrial channel modeling in the DD domain and derive a DSE-aware channel representation.
Then, we propose a novel frame structure that
utilizes the energy distribution of the received signal as prior knowledge to identify dominant DSE-aware channel components, while achieving low PAPR and OOBE.
Subsequently, we develop an efficient prior-aided iterative channel reconstruction (PAICR) algorithm to accurately reconstruct the complete satellite-terrestrial channel. The contributions of the paper are summarized as follows:
\begin{itemize}
\item[$\bullet$] We analyze the input-output relationship of the CP-OTFS-based LEO satellite communication system to maintain compatibility with existing networks and devices which commonly employ OFDM systems. In this system, time-frequency modulation and demodulation are performed using CP-OFDM, with the CP serving as a guard interval. Additionally, we derive a representation for the satellite-terrestrial channel in the DD domain to reveal the significant channel spreading caused by DSE, which destroys the DD-domain channel sparsity.
\item[$\bullet$] We design a OTFS frame structure to render the energy distribution of the DSE-aware channel explicitly observable at the receiver, which can be exploited as prior information for channel estimation.
Furthermore, by concentrating pilot symbols in the zeroth time slot of the DD domain, we achieve uniform energy distribution in the time domain, significantly reducing PAPR. To accommodate simple and practical rectangular pulses, we impose a continuity constraint at the stitching points of the time-domain waveform. Leveraging the unique transformation properties of OTFS, we derive DD-domain pilot symbols that fulfill this continuity constraint, thereby effectively suppressing OOBE.
\item[$\bullet$] We propose a PAICR algorithm to accurately reconstruct channels with non-sparse characteristics. The proposed approach adopts an iterative dominant component extraction and elimination strategy, where the strongest channel path is successively identified and removed from the received signal, thereby enabling progressive extraction of the channel parameters. A pilot-aided Doppler-domain energy observation mechanism is further exploited to provide prior knowledge for identifying dominant channel components, and a physically interpretable convergence criterion based on the evolution of the Doppler-domain energy distribution is designed to ensure reliable termination of the iterative process. Moreover, a Cramér-Rao lower bound (CRLB) is derived for the DSE-aware satellite-terrestrial channel model, providing a theoretical performance lower bound on channel estimation and serving as a benchmark for assessing our estimation approach.
\end{itemize}

The remainder of the paper is organized as follows. Section~\ref{section:2} presents the CP-OTFS-based LEO satellite communication system model. 
Section~\ref{section:3} describes the OTFS frame structure, detailing the arrangement of both transmit and receive symbols.
Our proposed PAICR algorithm tailored for the CP-OTFS-based LEO satellite system is presented in Section~\ref{section:4}.
Numerical results are illustrated in Section~\ref{section:5} for the proposed approach. Finally, a summary is provided in Section~\ref{section:6}.

\textit{Notations}: Scalars, vectors, and matrices are denoted by $x$, $\mathbf {x}$, and $\mathbf {X}$, respectively; 
$\delta\left( \cdot \right)$ denotes the Dirac delta function;
the superscripts $(\cdot)^*$ and $(\cdot)^H$ denote the complex conjugate and Hermitian transpose, respectively; $\lceil x \rceil$ and $\lfloor x \rfloor$ denote the smallest integer not less than $x$ and the largest integer not greater than $x$, respectively;  ${\left ( \cdot \right)_{M}}$ denotes mod-$M$ and $\langle x \rangle_N$ denotes $\left(x + \left\lfloor \frac{N}{2} \right\rfloor\right)_N - \left\lfloor \frac{N}{2} \right\rfloor$;
$\displaystyle \lim_{t \to a^-} f(t)$ and $\displaystyle \lim_{t \to a^+} f(t)$ denote the left- and right-hand limits of $f(t)$ as $t$ approaches $a$ from below and above, respectively; $\mathbb{E}\{\cdot\}$ denotes the expectation operator; $\mathcal{X}$ denotes a set; $\| \mathbf{x} \|_2$ denotes the Euclidean norm; $\Re\{\cdot\}$ denotes the real part of a complex quantity; and $\mathrm{tr}(\cdot)$ denotes the trace of a matrix.

\section{System Model}\label{section:2}
In this section, we first review the basic concepts of the CP-OTFS-based LEO satellite transmitter. Next, we analyze and derive the satellite-terrestrial channel model by taking into account the effects of DSE. Finally, we describe the processing of the signal received by the terrestrial receiver and characterize the input-output relationship of the system.

\subsection{LEO Satellite Transmitter Configuration}
To ensure that each symbol in the OTFS frame experiences nearly identical channel conditions, i.e., each symbol undergoes a constant channel gain, it is essential to map the symbol from the DD domain to the time-frequency (TF) domain. Therefore, the LEO satellite adopts an OTFS transmitter that applies the inverse symplectic finite Fourier transform (ISFFT)~\cite{7925924} to transform $Z=MN$ symbols from $X^{\mathrm {DD}}[k,l]$ to $X^{\mathrm {TF}}[n,m]$. 
The TF domain signal $X^{\mathrm {TF}}[n,m]$ obtained through ISFFT is given by

\begin{align} 
X^{\mathrm {TF}}[n,m]=\frac {1}{\sqrt{NM}} \sum _{k=\left\lceil -\frac{N}{2} \right\rceil}^{\left\lceil \frac{N}{2} \right\rceil - 1} \sum _{l=0}^{M-1} X^{\mathrm {DD}}[k,l] e^{-j 2 \pi \left ({\frac{m l}{M}-\frac{n k}{N}}\right)},
\label{eq:X_TF}
\end{align}
where $k=\left\lceil -\frac{N}{2} \right\rceil, \left\lceil -\frac{N}{2} \right\rceil + 1, \ldots,\left\lceil \frac{N}{2} \right\rceil - 1$, $l=0,1,\ldots,M-1$. 
Assume that the OFDM modulator adopts an easy-to-implement classical rectangular transmit waveform~\cite{Rectangular_Pulse}, i.e., 
\begin{align} 
g_{\mathrm {tx}}(t) \triangleq 
\begin{cases}
\frac{1}{\sqrt {T}}, & 0 \leq t \leq T_{\mathrm {sym}} \\ 
0, & {~\text{otherwise }}
\end{cases},
\end{align}
where $T$ denotes the time duration of an OFDM symbol without CP, and $T_{\mathrm {sym}}$ denotes the time duration of a complete OFDM symbol.
With the rectangular waveform, transforming the discrete signal $X^{\mathrm {TF}}[n,m]$ into the continuous transmit signal $s(t)$ yields
\begin{equation} 
\begin{aligned}
&\hspace {-.5pc} 
s(t)=\sum \limits_{m=0}^{M-1} \sum \limits _{n=0}^{N-1} X^{\mathrm{TF}}[n,m] e^{j 2 \pi m \Delta f\left({t-T_{\mathrm {cp}}-n T_{\mathrm {sym}}}\right)} \\&\qquad\qquad\qquad \qquad \qquad \qquad \qquad\times g_{\mathrm {tx}}\left ({t-n T_{\mathrm {sym}}}\right),
\label{eq:s(t)}
\end{aligned}
\end{equation}
where $m = 0,1,\ldots,M-1$, $n = 0,1,\ldots,N-1$, $\Delta f$ is the carrier spacing, $T_{\mathrm {cp}}=\frac{M_{\mathrm {cp}}T}{M}$ is the time duration of the CP, and $M_{\mathrm {cp}}$ is the length of CP. Here, we have $T=\frac{MT_{\mathrm {sym}}}{M+M_{\mathrm {cp}}}$.

\subsection{Satellite-Terrestrial Channel Analysis}
According to~\cite{hlawatsch2011wireless}, the time-variant channel impulse response $h(t,\tau)$ between the satellite and terrestrial receiver is defined as

\begin{equation} 
h(t,\tau)=\sum_{i=1}^{P}\mathcal G_{i}e^{j2\pi\nu_{i}t}\delta
\left(\tau-\tau_i(t)\right), 
\end{equation}
where $P$ is the number of propagation paths; $\tau_i(t) = \tau_i- \frac{v_i}ct$ is time-varying delay of the $i$-th path; $\nu_{i}=\frac{v_i}cf_c$ is Doppler shift at the carrier frequency $f_c$ of the $i$-th path; and $\mathcal G_{i}$, $\tau _{i}$ and $v_i$ respectively denote the attenuation, propagation delay and velocity of the $i$-th path. 
Substituting the time-varying delay, the channel impulse response $h(t,\tau)$ is rewritten as
\begin{equation} 
h(t,\tau)=\sum _{i=1}^{P}\mathcal G_{i}e^{j2\pi\nu_{i}t}\delta\left(\tau-(\tau_i-\frac{v_i}ct)\right). 
\label{eq:h(t,tau)}
\end{equation}

Based on~\eqref{eq:h(t,tau)}, the baseband equivalent response channel in the TF domain can be expressed as
\begin{align} 
H(t,f)=\sum _{i=1}^{P}\mathcal G_{i}e^{-j2\pi \tau _{i}f}e^{j2\pi \frac{\nu _{i}}{f_{c}}(f_{c}+f)t}. 
\end{align}

In contrast to the commonly-used channel models in prior LEO studies~\cite{introduction_data_aided1, vs_DSE1, 10982442, li2023channel}, the Doppler shift $e^{j2\pi \frac{\nu _{i}}{f_{c}}(f_{c}+f)t}$ introduces time-frequency coupling. Moreover, the phase shift induced by DSE accumulates as $e^{j2\pi \frac{\nu_{i}}{f_{c}}ft}$ within each OTFS symbol, and this phase variation depends on both the velocity and frame duration. 
While it can be considered negligible when the velocity is sufficiently low or the frame duration is sufficiently short, it becomes significant in LEO satellite systems. 
For instance, in a CP-OTFS-based LEO satellite system with typical parameters ($M=256$, $N=128$, and $v = 7562.2~\text{m/s}$), DSE produces a maximum offset of approximately $e^{j2\pi \frac{\nu _{i}}{f_{c}}\times M\Delta f\times NT}\approx e^{j1.65\pi}$. Such a considerable phase shift cannot be ignored and must be properly accounted for in system design.

\begin{figure*}[hb]
\centering
\hrulefill
\fontsize{9}{11}\selectfont
\begin{align*}
\label{eq:H_DD_ki_li}
&H_{k_i,l_i}^{\mathrm{DD}}\left[{k^{\prime}, l^{\prime}}\right] =
\mathcal G_{i}
e^{j2\pi\frac{k_i\left(l_i+M_{\mathrm{cp}}\right)}{N\left(M+M_{\mathrm{cp}}\right)}}
e^{j\pi\left(N-1\right)\left(M-1\right)\frac{M+M_{\mathrm{cp}}}{2Mp_{i}}}
e^{j\pi\left(N-1\right)\frac{k_i-k^{\prime}}{N}}
e^{j\pi\left(M-1\right)\frac{l^{\prime}-l_{i}}{M}}
\frac{\sin\pi N\left(\frac{k_{i}-k^{\prime}}{N}+(M-1)\frac{M+M_{\mathrm{cp}}}{2M\eta_i}\right)}{\sin\pi\left(\frac{k_{i}-k^{\prime}}{N}+(M-1)\frac{M+M_{\mathrm{cp}}}{2M\eta_i}\right)}
\\& \hspace {4.3pc}
\times \frac{\sin \pi M\left(\frac{l^{\prime}-l_{i}}{M}+(N-1)\frac{M+M_{\mathrm{cp}}}{2M\eta_i}\right)}{\sin\pi\left(\frac{l^{\prime}-l_{i}}{M}+(N-1)\frac{M+M_{\mathrm{cp}}}{2M\eta_i}\right)}.
\tag{13}
\end{align*}
\end{figure*}

\begin{figure*}[hb]
\centering
\hrulefill
\fontsize{9}{11}\selectfont
\begin{align*}
\label{eq:H_DD}
\tag{15}
&H^{\mathrm{DD}}\left[{k^{\prime}, l^{\prime}}\right] = \sum_{i=1}^{P} \mathcal G_{i}
e^{j2\pi\frac{k_i\left(l_i+M_{\mathrm{cp}}\right)}{N\left(M+M_{\mathrm{cp}}\right)}}
e^{j\pi\left(N-1\right)\left(M-1\right)\frac{M+M_{\mathrm{cp}}}{2M\eta_i}}
e^{j\pi\left(N-1\right)\frac{k_i-k^{\prime}}{N}}
e^{j\pi\left(M-1\right)\frac{l^{\prime}-l_{i}}{M}}
\frac{\sin\pi N\left(\frac{k_{i}-k^{\prime}}{N}+(M-1)\frac{M+M_{\mathrm{cp}}}{2M\eta_i}\right)}{\sin\pi\left(\frac{k_{i}-k^{\prime}}{N}+(M-1)\frac{M+M_{\mathrm{cp}}}{2M\eta_i}\right)}
\\&\hspace {4.2pc}
\times \frac{\sin \pi M\left(\frac{l^{\prime}-l_{i}}{M}+(N-1)\frac{M+M_{\mathrm{cp}}}{2M\eta_i}\right)}{\sin\pi\left(\frac{l^{\prime}-l_{i}}{M}+(N-1)\frac{M+M_{\mathrm{cp}}}{2M\eta_i}\right)}.
\end{align*} 
\end{figure*}

\setcounter{equation}{6}
To characterize the channel in the DD domain, we apply the symplectic finite Fourier transform (SFFT) to the time-variant frequency channel $H(t,f)$. 
Through the SFFT, $H(t,f)$ can be converted into a time-independent channel response
\begin{align} 
h(\tau,\nu)=\sum _{i=1}^{P}\mathcal G_{i}\left|\eta_i\right|e^{j2\pi \eta_i(\tau -\tau _{i})(\nu -\nu _{i})}, 
\label{eq:h(tau,nu)} 
\end{align}
where $\eta_i = \frac{f_c}{\nu_i}$ is introduced for notational simplicity, $\nu _{i} = \frac {k_{i} + \widetilde{k}_i}{NT_{\mathrm {sym}}}$ is Doppler tap for the $i$-th path and $\tau _{i}=\frac {l_{i} }{M\Delta f}$ is the delay tap for the $i$-th path. $k_{i}$ and $l_i$ are integers and represent the indexes of Doppler and delay taps. The real number $\widetilde{k}_i$, whose value range is $[-\frac{1}{2}, \frac{1}{2})$ , is defined as the fractional Doppler. 
The fractional delay is not considered since the delay-axis resolution $\frac {1}{M\Delta f}$ is sufficiently fine to map each path delay to an integer delay tap in wide-band LEO satellite systems (up to 30 MHz in the S-band~\cite{3GPP})~\cite{Tse_Viswanath_2005}.
Assume that the DD-domain channel is sparse, which is typically modeled as $h(\tau,\nu)=\sum _{i=1}^{P}\mathcal G_{i} \delta (\tau -\tau _{i}) \delta (\nu -\nu _{i})$~\cite{10856399, 10666708, 10679987, 10752433, 10916513}. However, as shown in~\eqref{eq:h(tau,nu)}, DSE destroys the sparsity of the satellite-terrestrial channel in the DD domain.



\subsection{Terrestrial Receiver Processing}
The transmit signal $s(t)$ is given in~\eqref{eq:s(t)}.
After the transmission over the satellite-terrestrial channel, the received signal is shown below
\begin{equation} 
r(t)=\iint h(\tau,\nu)s(t-\tau)e^{j2\pi \nu t}d\tau d\nu + n(t), 
\end{equation}
where $n(t)$ denotes the additive Gaussian noise.

At the receiver, the matched filter computes the cross-ambiguity function between the received waveform $g_{\mathrm{rx}}$ and the received signal $r(t)$ as
\begin{equation} 
A_{g_{\mathrm {rx}}, r}(t, f) = \int g_{\mathrm {rx}}^{*}\left ({t^{\prime}-t}\right) r\left ({t^{\prime}}\right) e^{-j 2 \pi f\left ({t^{\prime}-T_{\mathrm {cp}}}\right)} d t^{\prime},
\label{eq:A}
\end{equation}
where $g_{\mathrm{rx}}$ is defined as
\begin{align*} 
g_{\mathrm {rx}}(t) \triangleq 
\begin{cases}
\frac{1}{\sqrt {T}}, & T_{\mathrm{cp}} \leq t \leq T_{\mathrm {sym}} \\ 
0, & {~\text {otherwise }}
\end{cases}.
\end{align*}

Based on~\eqref{eq:A}, the received symbols are sampled as
\begin{equation} 
Y^{\mathrm {TF}}[n, m]=\left.{A_{g_{\mathrm{rx}}, r}(t, f)}\right |_{t=n T_{\mathrm{sym}}, f=m \Delta f}.
\label{eq:Y_TF} 
\end{equation}

Finally, the SFFT is applied to $Y^{\mathrm{TF}}[n, m]$ to obtain the received signal samples in the DD domain, i.e.,
\begin{equation} 
Y^{\mathrm {DD}}[k,l]=\frac {1}{\sqrt {N M}} \sum _{n=0}^{N-1} \sum _{m=0}^{M-1} Y^{\mathrm {TF}}[n, m] e^{j 2 \pi \left ({\frac {m l}{M}-\frac {n k}{N}}\right)}.
\label{eq:Y_DD} 
\end{equation}

According to the signal model above, the input-output relationship of the CP-OTFS-based LEO satellite system is derived in the following proposition.

\begin{prop}
For the CP-OTFS-based LEO satellite system, the input-output relationship is described as
\begin{equation} 
\begin{aligned}
&\hspace {-.9pc} 
Y^{\mathrm {DD}}[k, l]=\frac {1}{N M} \sum_{i=1}^{P} \sum _{k^{\prime}=\left\lceil -\frac{N}{2} \right\rceil}^{\left\lceil \frac{N}{2} \right\rceil - 1} \sum _{l^{\prime}=0}^{M-1} \theta_{k_i}(l,l^{\prime}) H_{k_i,l_i}^{\mathrm {DD}}\left [{k^{\prime}, l^{\prime}}\right] \\
&\qquad\qquad\times X^{\mathrm {DD}}\left [{\left \langle{ k-k^{\prime}}\right \rangle _{N}, \left ({l-l^{\prime}}\right )_{M} } \right] + V^{\mathrm {DD}}[k, l],
\label{eq:Y_DD1} 
\end{aligned}
\end{equation}
where $H_{k_i,l_i}^{\mathrm {DD}}\left [{k^{\prime},l^{\prime}}\right]$ is formulated as~\eqref{eq:H_DD_ki_li}, shown at the bottom of the page; $\theta_{k_i}(l,l^{\prime}) = e^{j 2 \pi \frac {k_i{\left ({l-l^{\prime}}\right )_{M}}}{N(M+M_{\mathrm{cp}})}}$ represents the phase rotation factor; and $V^{\mathrm {DD}}[k, l]$ is the additive noise in the DD domain.
\end{prop}

\begin{IEEEproof}
    See Appendix.
\end{IEEEproof}

\setcounter{equation}{13}

Notably, the phase rotation factor $\theta_{k_i}(l,l^{\prime})$ in~\eqref{eq:Y_DD1} poses a practical challenge for receivers, as its estimation involves inferring all Doppler taps (i.e., determining all indices $k_i$) for each propagation path. 
However, when the Doppler dimension $N$ is sufficiently large, the signal model in~\eqref{eq:Y_DD1} can be approximated as~\cite{Y_DD_N_INF2}
\begin{equation} 
\begin{aligned}
&\hspace {-.9pc} 
Y^{\mathrm {DD}}[k,l] \approx
\frac {1}{N M} \sum _{k^{\prime}=\left\lceil -\frac{N}{2} \right\rceil}^{\left\lceil \frac{N}{2} \right\rceil - 1} \sum _{l^{\prime }=0}^{M-1} \theta_{k^{\prime}}(l,l^{\prime}) H^{\mathrm {DD}}\left [{k^{\prime }, l^{\prime }}\right] \\
&\qquad\qquad\times X^{\mathrm {DD}}\left [{\left \langle{ k-k^{\prime }}\right \rangle _{N}, \left ({l-l^{\prime}}\right )_{M} } \right] + V^{\mathrm {DD}}[k, l],
\label{eq:Y_DD_inf}
\end{aligned}
\end{equation}
where $H^{\mathrm {DD}}\left [{k^{\prime }, l^{\prime }}\right]$ is formulated as~\eqref{eq:H_DD}, shown at the bottom of the page; the phase potation factor turns into $\theta_{k^{\prime}}(l,l^{\prime}) = e^{j 2 \pi \frac {k^{\prime}{\left ({l-l^{\prime}}\right )_{M}}}{N(M+M_{\mathrm{cp}})}}$.
We observe from~\eqref{eq:Y_DD_inf} that this mathematical operation yields an approximate simplified expression for the original phase rotation factor $\theta_{k_i}(l,l^{\prime})$, now denoted as $\theta_{k^{\prime}}(l,l^{\prime})$, with only minor performance degradation, as will be confirmed by simulation results presented later.

\section{OTFS Frame Structure Arrangement}\label{section:3}


In this section, we first investigate the pilot arrangement strategy that enables the received signal to intuitively reveal the energy distribution of the DSE-aware channel. Subsequently, we optimize the placement of pilot symbols to reduce the PAPR of the time-domain waveform. We then theoretically derive the structural conditions that pilot symbols must satisfy in the DD domain to ensure low OOBE. Finally, by integrating these three design principles, we propose a novel OTFS frame structure that facilitates efficient channel estimation.

\subsection{DSE-Aware Design}
Since DSE leads to severe power leakage of the channel in the DD domain, our goal is to visually capture the structural characteristics of the dominant channel energy from the received signal, thus providing strong initial information for subsequent channel estimation.
It is worth noting that pilot symbols are known to the receiver and can therefore be regarded as deterministic probes for revealing channel energy dispersion in the DD domain. Accordingly, this subsection focuses on the positions of pilot symbols in the OTFS frame, while the interference from data symbols is ignored. For analytical clarity, we assume a pilot-only OTFS frame, in which all non-pilot positions are set to zero. It can be expressed as
\setcounter{equation}{15}
\begin{equation}
X^{\mathrm{DD}}[k, l] = \sum_{(k_p^{(i)}, l_p^{(i)}) \in \mathcal{P}} \delta[k - k_p^{(i)}] \cdot \delta[l - l_p^{(i)}],
\label{eq:pilot_position}
\end{equation}
where \( \mathcal{P} \) denotes the set of pilot indices; $k_p^{(i)}$ and $l_p^{(i)}$ denote the positions of the pilots in the Doppler dimension and the delay dimension, respectively. 
In the subsequent frame design, pilot and data symbols are separated by guard intervals, and pilot symbols are typically transmitted at higher power than data symbols. Accordingly, the analytical assumption, which ignores data interference, serves as a reasonable abstraction rather than a restrictive condition.

By substituting~\eqref{eq:pilot_position} into~\eqref{eq:Y_DD_inf}, we obtain
\begin{align}
\hspace {-1pc} 
&Y^{\mathrm{DD}}[k, l] \nonumber\\
&\approx \frac{1}{NM} \sum_{(k_p^{(i)}, l_p^{(i)}) \in \mathcal{P}} \sum _{k^{\prime}=\left\lceil -\frac{N}{2} \right\rceil}^{\left\lceil \frac{N}{2} \right\rceil - 1} \sum_{l^{\prime}=0}^{M-1} 
\theta_{k^{\prime}}(l, l^{\prime}) H^{\mathrm{DD}}[k^{\prime}, l^{\prime}] \nonumber\\
&\quad \times \delta[\langle k - k^{\prime} - k_p^{(i)} \rangle_N] \cdot 
\delta[ ( l - l^{\prime} - l_p^{(i)} )_M] + V^{\mathrm{DD}}[k, l] \nonumber\\
&= \frac{1}{NM} \sum_{(k_p^{(i)}, l_p^{(i)}) \in \mathcal{P}} \theta_{\langle k - k_p^{(i)} \rangle_N}\left(l, ( l - l_p^{(i)} )_M\right) \nonumber \\
&\quad \times H^{\mathrm{DD}}\left[\langle k - k_p^{(i)} \rangle_N, ( l - l_p^{(i)} )_M \right] + V^{\mathrm{DD}}[k, l].
\label{eq:ydd_observation}
\end{align}

From~\eqref{eq:ydd_observation}, it can be observed that each pilot symbol located at \( (k_p^{(i)}, l_p^{(i)}) \in \mathcal{P} \) contributes a spatially shifted replica of the DSE-aware channel.
The overall observation $Y^{\mathrm{DD}}[k, l]$ is therefore determined by both the shift and the superposition effects induced by the pilot arrangement.
Since the DSE causes the channel to spread along both the Doppler and delay axes around strong paths~\cite{10103827}, obtaining dominant information in either dimension allows for inferring the approximate distribution region of the non-zero channel coefficients. If all pilot symbols lie on a single delay column or a single Doppler row, circular shifts are introduced along one dimension, while no mutual superposition occurs in that dimension among different pilot contributions. In this case, the dominant path locations remain distinguishable, and the reliability of subsequent path localization is largely preserved. In contrast, when pilot symbols are distributed over multiple rows and columns, multiple shifted replicas are superimposed, leading to energy overlap and structural ambiguity. This superposition effect reduces the reliability of dominant path identification and increases the complexity of subsequent iterative estimation and interference cancellation procedures.
Therefore, pilot symbols are preferably placed with either $k_p^{(i)}=0$ or $l_p^{(i)}=0$, which can be regarded as a special case of the single-row or single-column arrangement, ensuring a zero circular shift in at least one dimension, thereby enabling a more intuitive visualization of the channel energy distribution.

\subsection{PAPR-Reduced Design}
The PAPR of the time-domain continuous signal in ~\eqref{eq:s(t)} can be defined as
\begin{equation}
PAPR =\frac{\max\limits_{0\le t\le NT_\mathrm{sym}}\left|s\left(t\right)\right|^2}{\mathbb{E}\{|s(t)|^2\}}.
\end{equation}
As shown in~\cite{tellado2005multicarrier}, in order to facilitate analysis and optimization, the continuous-time waveform $s(t)$ is typically oversampled by a factor of at least 4, so that the PAPR of the resulting discrete-time signal can closely approximate that of the actual continuous-time waveform. Therefore, by reducing the PAPR of the oversampled discrete sequence, it is possible to lower the actual PAPR of the continuous-time waveform to some extent.
To derive the condition for the pilot to achieve reduced PAPR in the DD domain, we first sample \eqref{eq:s(t)} as
\begin{equation}
\begin{aligned}
&s(uT_\mathrm{s}) = \sum_{n=0}^{N-1}\sum_{m=0}^{M-1}X^\mathrm {TF}[n,m]e^{j2\pi \Delta f(uT_\mathrm{s}-T_\mathrm{cp}-nT_\mathrm{sym})}\\
&\qquad \qquad \times g_\mathrm{tx}(uT_\mathrm{s}-nT_\mathrm{sym}),
\label{eq:s_uT}
\end{aligned}
\end{equation}
where $u = r+z(M+M_\mathrm{cp})$, $r = 0, 1,\dots,M+M_\mathrm{cp}-1$, $z = 0,1,\dots N-1$, and $T_\mathrm{s}$ denotes the sampling period. Substituting~\eqref{eq:X_TF} into~\eqref{eq:s_uT} yields
\begin{equation}
\begin{aligned}
&\hspace {-2.5pc} s\left(\left(r+z\left(M+M_\mathrm{cp}\right)T_\mathrm{s}\right)\right) \\
&\hspace {-2.5pc} = \frac{1}{\sqrt{MN}}\sum_{n=0}^{N-1}\sum_{m=0}^{M-1} \sum _{k=\left\lceil -\frac{N}{2} \right\rceil}^{\left\lceil \frac{N}{2} \right\rceil - 1} \sum_{l=0}^{M-1}X^\mathrm {DD}[k,l]e^{j2\pi \frac{nk}{N}} \\ 
& \hspace {-1.5pc} \times e^{j2\pi \frac{m}{M}\left(r-M_\mathrm{cp}-l+\left(z-n\right)\left(M+M_\mathrm{cp}\right)\right)}\\
& \hspace {-1.5pc} \times g_\mathrm{tx}\left(\left(r+\left(z-n\right)\left(M+M_\mathrm{cp}\right)\right)T_\mathrm{s}\right).
\label{eq:s_discrete1}
\end{aligned}
\end{equation}

It can be observed from~\eqref{eq:s_discrete1} that when $z = n$, the sampled time instant lies within the support of the rectangular pulse. Therefore,~\eqref{eq:s_discrete1} can be rewritten as
\begin{align}
&s\left(\left(r+z\left(M+M_\mathrm{cp}\right)T_\mathrm{s}\right)\right) \nonumber\\
&= \frac{1}{\sqrt{MN}}\sum_{m=0}^{M-1} \sum _{k=\left\lceil -\frac{N}{2} \right\rceil}^{\left\lceil \frac{N}{2} \right\rceil - 1} \sum_{l=0}^{M-1}X^\mathrm {DD}[k,l]e^{j2\pi \frac{zk}{N}} e^{j2\pi \frac{m}{M}\left(r-M_\mathrm{cp}-l\right)} \nonumber\\
&= \sqrt{\frac{M}{N}} \sum _{k=\left\lceil -\frac{N}{2} \right\rceil}^{\left\lceil \frac{N}{2} \right\rceil - 1} X^{\mathrm{DD}}\left[k,\rho_1\right]e^{j2\pi \frac{zk}{N}},
\label{eq:s_discrete2}
\end{align}
where $\rho_1 = \langle r-M_{\mathrm{cp}} \rangle _M$. Then the PAPR can be calculated as~\cite{8701706}
\begin{equation}
PAPR =\frac{ \max\limits_{\rho_1,z}\left|s\left(\left(r+z\left(M+M_\mathrm{cp}\right)T_\mathrm{s}\right)\right)\right|^2}{P_{\text{avg}}},
\end{equation}
where $P_{\text{avg}} = \frac{1}{(M+M_{\mathrm{cp}})N}\sum\limits_{r=0}^{M+M_{\mathrm{cp}}-1}\sum\limits_{z=0}^{N-1}\mathbb{E}\{|s((r+$
$z(M+M_\mathrm{cp})T_\mathrm{s}))|^2\}$. The upper bound of PAPR can be calculated from the Parseval's theorem and the Cauchy-Schwarz inequality theorem as follows:
\begin{equation}
PAPR \le \frac{MN \max\limits_{k,\rho_1}\left|X^\mathrm {DD}[k,\rho_1]\right|^2}{M\sigma_x^2} = \frac{N \max\limits_{k,\rho_1}\left|X^\mathrm {DD}[k,\rho_1]\right|^2}{\sigma_x^2},
\label{PAPR1}
\end{equation}
where $\sigma_x^2 = \mathbb{E}\{|X^\mathrm {DD}[k,\rho_1]|^2\}$. Based on the analysis of DSE-aware pilot allocation, we choose to amplify the pilot power at $k=0$ by a factor of $\alpha_0$ for all $0 \le l \le M-1$. Therefore, the PAPR is expressed as
\begin{equation}
\begin{aligned}
&PAPR \\
&=\frac{\max \limits_{\rho_1,z}\left|\sqrt{\frac{M}{N}}\left(\alpha_0X^{\mathrm{DD}}\left[0,\rho_1\right]+\sum\limits_{\substack{k=\left\lceil -\frac{N}{2} \right\rceil \\ k \neq 0}}^{\left\lceil \frac{N}{2} \right\rceil - 1} X^{\mathrm{DD}}\left[k,\rho_1\right]e^{j2\pi \frac{zk}{N}}\right)\right|^2}{P_{\text{avg}}},
\label{PAPR'}
\end{aligned}
\end{equation}
with
\begin{align}
&P_{\text{avg}}=\frac{1}{(M+M_{\mathrm{cp}})N}\sum\limits_{r=0}^{M+M_{\mathrm{cp}}-1}\sum\limits_{z=0}^{N-1} \nonumber\\
& \mathbb{E}\left\{\left|\sqrt{\frac{M}{N}}\left(\alpha_0X^{\mathrm{DD}}\left[0,\rho_1\right]+\sum\limits_{\substack{k=\left\lceil -\frac{N}{2} \right\rceil \\ k \neq 0}}^{\left\lceil \frac{N}{2} \right\rceil - 1}X^{\mathrm{DD}}\left[k,\rho_1\right]e^{j2\pi \frac{zk}{N}}\right)\right|^2\right\}.
\label{Pavg'}
\end{align}
Substituting~\eqref{Pavg'} into~\eqref{PAPR'}, we get
\begin{align}
&PAPR \nonumber \\
&\approx \frac{\max\limits_{\rho_1,z}\left|\sqrt{\frac{M}{N}}\alpha_0X^{\mathrm{DD}}\left[0,\rho_1\right]\right|^2}{\frac{1}{(M+M_{\mathrm{cp}})N}\sum\limits_{r=0}^{M+M_{\mathrm{cp}}-1}\sum\limits_{z=0}^{N-1} \mathbb{E}\left\{\left|\sqrt{\frac{M}{N}}\alpha_0X^{\mathrm{DD}}\left[0,\rho_1\right]\right|^2\right\}} \nonumber \\
&\le \frac{\frac{M}{N}\alpha_0^2 \max\limits_{k,\rho_1}\left|X^\mathrm {DD}[k,\rho_1]\right|^2}{\frac{M}{N}\alpha_0^2\sigma_x^2} 
= \frac{\max \limits_{k,\rho_1}\left|X^\mathrm {DD}[k,\rho_1]\right|^2}{\sigma_x^2}.
\label{PAPR2}
\end{align}
From \eqref{PAPR2}, we conclude that the PAPR is reduced by a factor of $N$ compared to~\eqref{PAPR1}.

\begin{figure*}[t]
\centering
\subfloat[]{\label{fig:pilotT}
		\includegraphics[width=0.42\textwidth]{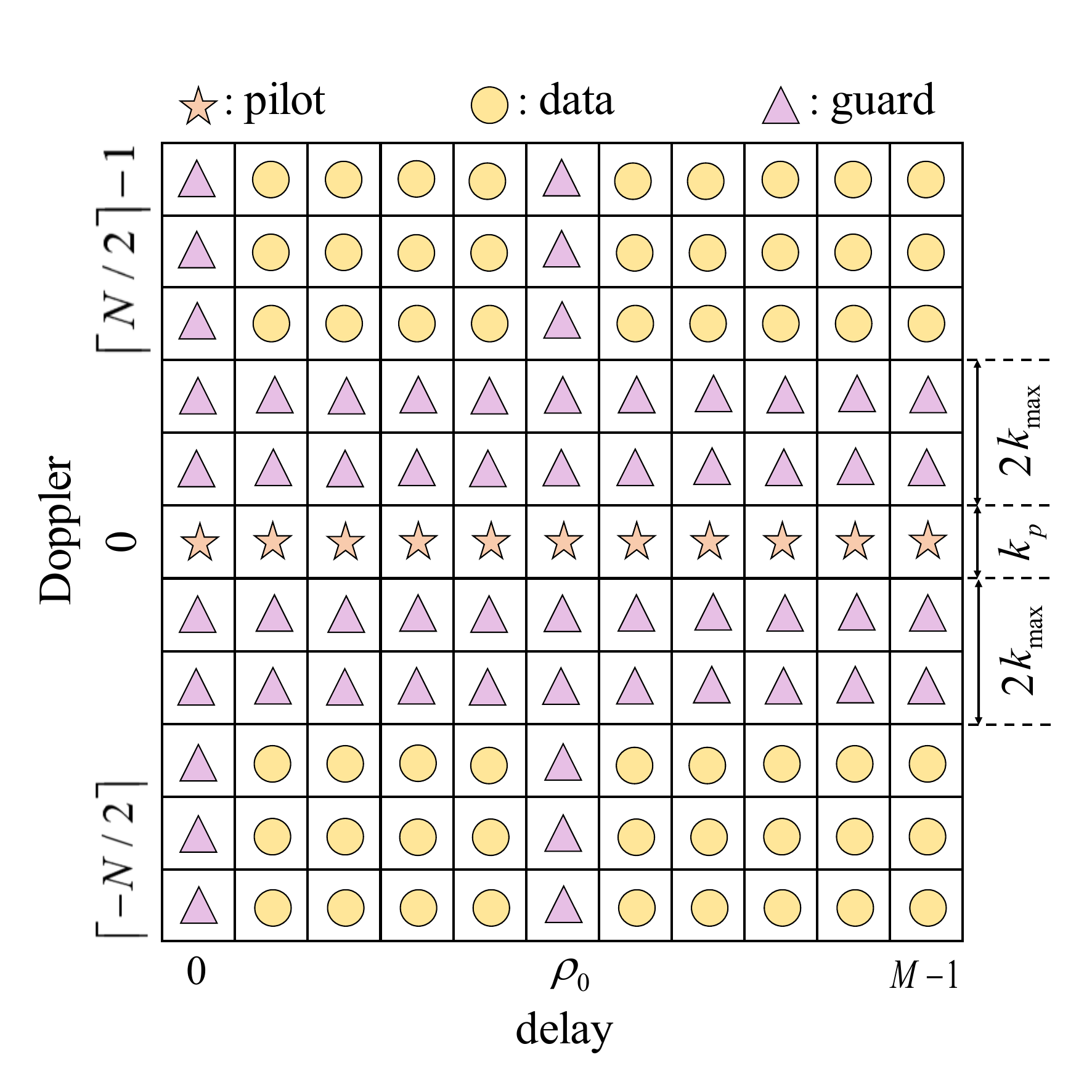}}
\hspace{7 mm}
\subfloat[]{\label{fig:pilotR}
		\includegraphics[width=0.42\textwidth]{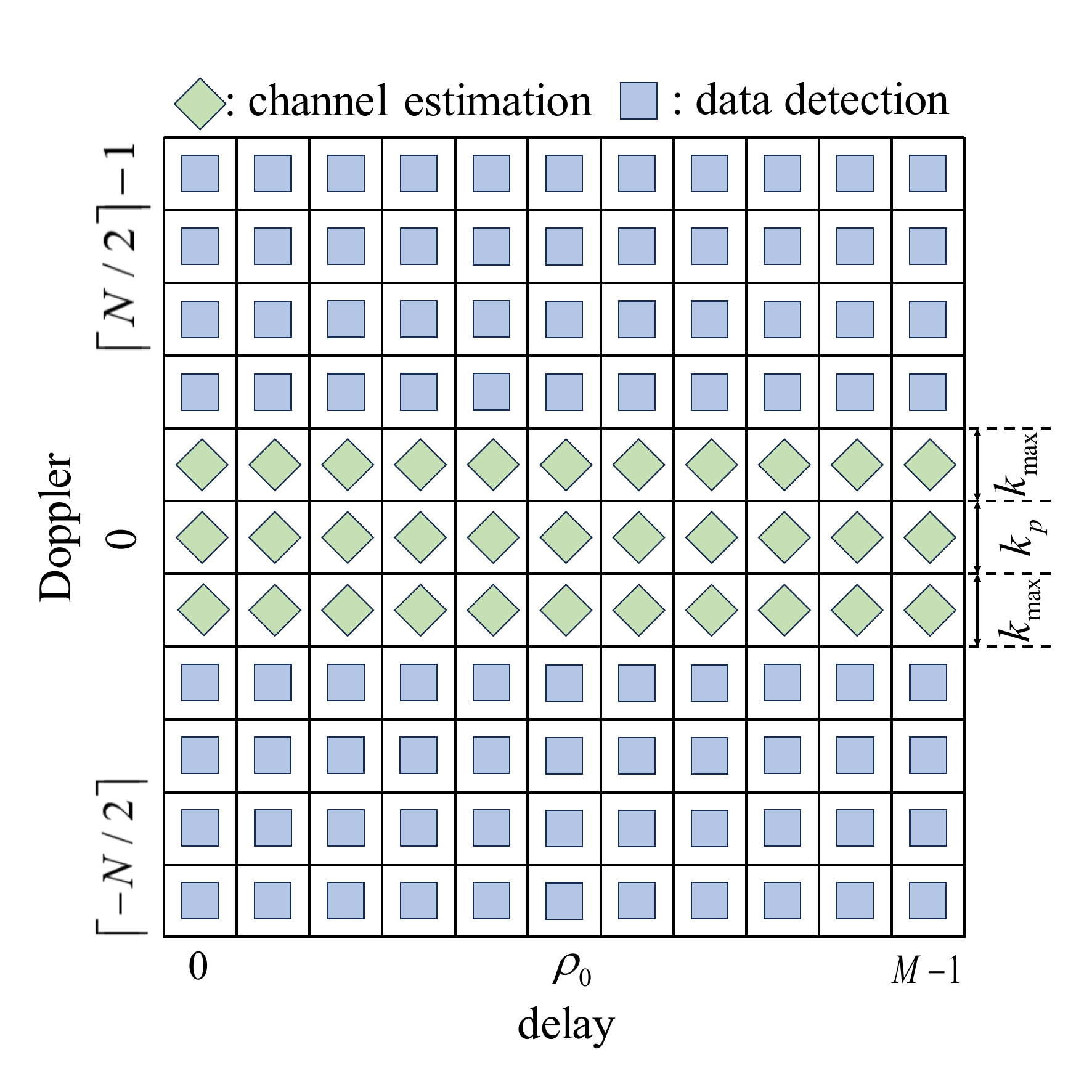}}
\caption{DD-domain symbol arrangement design. (a) Transmit symbol arrangement. (b) Receive symbol arrangement.}
\label{fig:pilot}
\end{figure*}

\subsection{OOBE-Reduced Design}
In practical communication systems, rectangular pulses are widely employed for their simplicity and operational feasibility~\cite{9920179}. However, their abrupt transitions at signal boundaries introduce discontinuities, which, according to Fourier theory, give rise to substantial high-frequency components and induce out-of-band spectral leakage. To mitigate this issue while preserving the rectangular pulse structure, we enforce continuity at the symbol boundaries. This approach effectively suppresses high-frequency spectral components and reduces OOBE.

To ensure the continuity of $s(t)$ at all times, particularly at the symbol boundary $t = nT_\mathrm{sym}$, we introduce the following continuity condition:
\begin{equation} 
\lim_{t \to nT_\mathrm{sym}^-} s(t) = \lim_{t \to nT_\mathrm{sym}^+} s(t).
\label{eq:continuity_cond}
\end{equation}
According to~\eqref{eq:s(t)}, the continuity condition in \eqref{eq:continuity_cond} can be rewritten as
\begin{align} 
&\sum_{m=0}^{M-1} X^\mathrm {TF}[n-1,m] e^{j2\pi m \Delta f(T_{\mathrm{sym}}-T_{\mathrm{cp}})} g_{\mathrm{tx}}(T_{\mathrm{sym}}) \nonumber \\
=& \sum_{m=0}^{M-1} X^\mathrm {TF}[n,m]e^{-j2\pi m\Delta f T_{\mathrm {cp}}} g_{\mathrm{tx}}(0).
\label{eq:cc_TF1}
\end{align}
Further,~\eqref{eq:cc_TF1} can be translated into
\begin{align}
&\sum_{m=0}^{M-1} X^\mathrm {TF}[n-1,m] 
= \sum_{m=0}^{M-1} X^\mathrm {TF}[n,m]e^{-j2\pi m\Delta f T_\mathrm {cp}}.
\label{eq:cc_TF}
\end{align}

Converting the continuity condition in~\eqref{eq:cc_TF} to the DD domain yields
\begin{equation}
\begin{aligned}
&\hspace {-0.5pc} 
\sum _{k=\left\lceil -\frac{N}{2} \right\rceil}^{\left\lceil \frac{N}{2} \right\rceil - 1}\sum_{l=0}^{N-1}\sum_{m=0}^{M-1}X^\mathrm {DD}[k,l]e^{j2\pi(n-1)\frac{k}{N}}e^{-j2\pi m\frac{l}{M}} \\
&=\sum _{k=\left\lceil -\frac{N}{2} \right\rceil}^{\left\lceil \frac{N}{2} \right\rceil - 1}\sum_{l=0}^{N-1}\sum_{m=0}^{M-1}X^\mathrm {DD}[k,l]e^{j2\pi n\frac{k}{N}}e^{-j2\pi m \left( \frac{l}{M} + \Delta f T_\mathrm{cp}\right )} \\
& \sum _{k=\left\lceil -\frac{N}{2} \right\rceil}^{\left\lceil \frac{N}{2} \right\rceil - 1}X^\mathrm {DD}[k,0]e^{j2\pi(n-1)\frac{k}{N}} = \sum _{k=\left\lceil -\frac{N}{2} \right\rceil}^{\left\lceil \frac{N}{2} \right\rceil - 1}X^\mathrm {DD}[k,\rho_0]e^{j2\pi n\frac{k}{N}},
\label{eq:cc_DD}
\end{aligned}
\end{equation}
where $\rho_0 = \langle -M_\mathrm{cp}\rangle_M$. Based on~\eqref{eq:cc_DD}, the pilot arrangement conditions that reduce OOBE are derived as
\begin{align} 
X^\mathrm {DD}[k,0] = X^\mathrm {DD}[k,\rho_0]e^{j2\pi \frac{k}{N}}, {\left\lceil -\frac{N}{2} \right\rceil} \le k \le {\left\lceil \frac{N}{2} \right\rceil - 1}. 
\label{eq:OOBE_condition}
\end{align}
Note that this condition holds for all $1 \le n \le N-1$, and it provides a sufficient condition for reducing OOBE. When $n=0$ is also included, it becomes a necessary and sufficient condition.

Building on the analyses in the preceding subsections on pilot, we now proceed to design the OTFS frame structure. Specifically, pilot symbols are placed at positions $[0,l]$ $(0 \le l \le M-1)$, with identical pilot symbols assigned to positions $[0,0]$ and $[0,\rho_0]$. To avoid interference between pilot and data symbols, the positions $[k,0]$ and $[k,\rho_0]$ $\left({\left\lceil -\frac{N}{2} \right\rceil} \le k \le {\left\lceil \frac{N}{2} \right\rceil - 1}, k \neq 0\right)$ are set to zero. This is because placing other pilot symbols that satisfy condition~\eqref{eq:OOBE_condition} at these positions would result in interference between pilot and data symbols. Furthermore, to enhance interference mitigation, a guard interval is introduced along the Doppler dimension, with a length equal to $2k_{\max}$~\cite{8671740}, where $k_{\max}$ denotes the Doppler taps corresponding to the largest Doppler.

As illustrated in Fig.~\ref{fig:pilot}\subref{fig:pilotT}, the symbols at the transmitter are arranged as
\begin{align} 
x[k,l] = {\begin{cases}
x_p[k,l], & k=0, 0 \leq l \leq M-1\\ 
0, & {\left\lceil -\frac{N}{2} \right\rceil} \le k \le {\left\lceil \frac{N}{2} \right\rceil - 1}, k \neq 0, l=0,\rho_0 \\
& \text{and}~k_p - 2k_{\max}\le k \le k_p + 2k_{\max}, \\
& 0 \leq l \leq M-1\\
x_d[k,l], & \text{otherwise.} 
\end{cases}} 
\label{eq:pilot_arrangement}
\end{align}
where $x_p[k,l]$ and $x_d[k,l]$ denote the pilot and data symbols; 
And $k_p$ is the positions of the pilots in the Doppler dimension.

At the receiver, the received symbols $y[k,l],k_p-k_{\max} \leq k \leq k_p+k_{\max},0 \leq l \leq M-1$ are used for channel estimation, while the remaining received symbols $y[k,l]$ are used for data detection, as shown in Fig.~\ref{fig:pilot}\subref{fig:pilotR}.

\section{Non-Sparse Satellite-Terrestrial Channel Reconstruction}\label{section:4}

In this section, we formulate the channel estimation problem for the CP-OTFS-based LEO satellite system.
To address this problem, we propose a PAICR algorithm, where dominant channel parameters are iteratively extracted from a coarse estimate using prior information derived from received signal energy observation, and the complete channel is subsequently reconstructed with convergence determined by an energy variation criterion.

\subsection{Problem Formulation}
With the considered frame structure arrangement scheme, the channel estimation task can be carried out by using the known symbols.

By utilizing the pilot and guard symbols in an OTFS frame,~\eqref{eq:Y_DD_inf} can be rewritten as 
\begin{equation} 
\begin{aligned}
&\hspace {-.9pc} Y^{\mathrm {DD}}[k, l]=\frac{1}{N M} \sum_{k^{\prime }=-k_{\max}}^{k_{\max}} \sum_{l^{\prime }=0}^{M-1}\theta_{k^{\prime}}(l,l^{\prime}) H^{\mathrm {DD}}\left[{k^{\prime }, l^{\prime }}\right] \\
&\qquad\qquad\times X^{\mathrm{DD}}\left[{\left \langle{ k-k^{\prime }}\right \rangle_{N}, \left ({l-l^{\prime}}\right )_{M}}\right] + V^{\mathrm{DD}}[k, l],
\label{eq:Y_DDP} 
\end{aligned}
\end{equation} 
To facilitate channel estimation, the system model in~\eqref{eq:Y_DDP} can be expressed in vectorized from as
\begin{equation} 
\mathbf {y}_{p}=\boldsymbol{\Phi }_{p}(\boldsymbol k, \boldsymbol l) \boldsymbol{\mathcal G}  +\mathbf {v}_{p},
\label{eq:hp_for_channel_estimation} 
\end{equation}
where vector $\mathbf {y}_{p} \in \mathbb{C}^{M_TN_T \times 1}$, $\boldsymbol{\mathcal G} \in \mathbb{C}^{P \times 1}$ and $\mathbf {v}_p \in \mathbb{C}^{M_TN_T \times 1}$ as the vectorized $\mathbf Y^{\mathrm {DD}}$, $\mathcal G_i$ and $\mathbf V^{\mathrm {DD}}$ with $M_T=M$ and $N_T=2k_{\max}+1$. The measurement matrix $\boldsymbol{\Phi }_{p}(\boldsymbol k, \boldsymbol l) \in \mathbb{C}^{M_TN_T \times P}$ can be expressed as
\begin{equation} 
\begin{aligned}
\boldsymbol{\Phi }_{p}(\boldsymbol k, \boldsymbol l)=\left[ \boldsymbol{\phi}_{p}(k_1,l_1), \boldsymbol{\phi}_{p}(k_2,l_2), \dots,\boldsymbol{\phi}_{p}(k_P,l_P)\right],
\label{eq:phi_p}
\end{aligned}
\end{equation}
where $\boldsymbol{\phi }_{p}(k_i,l_i)\in \mathbb{C}^{M_TN_T \times 1}$ and its $(Mk+l)$-th entry is given by
\begin{equation} 
\begin{aligned}
&\{\boldsymbol {\phi} _{p}(k_i,l_i)\}_{kM+l}=\frac{1}{NM} \sum_{k^{\prime }= -k_{\max}}^{k_{\max}} \sum_{l^{\prime }=0}^{M-1} \psi_i \left[k',l' \right] \\
&\qquad\qquad\qquad \times \theta_{k^{\prime}}(l,l^{\prime}) 
X^{\mathrm {DD}}\left [{\left \langle{ k-k^{\prime }}\right \rangle _{N}, \left ({l-l^{\prime}}\right )_{M}}\right]
\end{aligned}
\end{equation}
with
\begin{equation} 
\psi_i \left[{k^{\prime}, l^{\prime}}\right] = \frac{H_{k_i,l_i}^{\mathrm{DD}}\left[{k^{\prime}, l^{\prime}}\right]}{\mathcal G_{i}}.
\end{equation}


Owing to the presence of unknown delay and Doppler shifts, the formulation in~\eqref{eq:phi_p} remains nonlinear. To address this issue, a first-order Taylor expansion is applied to linearize the estimation problem. Let $\overline{\boldsymbol{k}}_g = \{\overline{k}_0, \overline{k}_1, \dots, \overline{k}_{N_{\nu}-1}\}$ denote a uniform sampling grid over the Doppler range $[-k_{\max}, k_{\max}]$ with a virtual Doppler resolution of $r_{\nu} = \frac{2k_{\max}}{N_{\nu}}$, where $N_{\nu}$ denotes the virtual grid size in the Doppler domain. Let $\overline{\boldsymbol{l}}_g = \{\overline{l}_0, \overline{l}_1, \dots, \overline{l}_{M_{\tau}-1}\}$ denote a uniform sampling grid over the delay range $[0, l_{\max}]$. Since fractional delays are not considered, $M_{\tau}=l_{\max}+1$ is defined as the virtual grid size in the delay domain, and the virtual delay resolution is set to $r_\tau = 1$. Based on the constructed virtual sampling grid, a first-order linear approximation can be performed as follows:
\begin{equation} 
\boldsymbol{\phi}_p({k}_i,l_i) \approx \boldsymbol{\phi}_p(\overline{k}_i,l_i)+\boldsymbol{\phi}_{p,\nu}'(\overline{k}_i,l_i)\kappa_i,
\end{equation} 
where $\overline{k}_i$ denotes the nearest grid point of ${k}_i$, and $\kappa_i = {k}_i-\overline{k}_i \in [-\frac{r_{\nu}}{2},\frac{r_{\nu}}{2}]$ represents the off-grid component. Vector $\boldsymbol{\phi}_{p,\nu}'(k_i,l_i)=\frac{\partial{\boldsymbol{\phi}_p(k_i,l_i)}}{\partial{k_i}} \in \mathbb{C}^{M_TN_T \times 1}$ denotes the first-order derivative of $\boldsymbol{\phi}_p(k_i,l_i)$ with respect to $k_i$. 
\begin{figure*}[ht]
\centering
\subfloat[]{\label{subfig:frac}
		\includegraphics[scale=0.4]{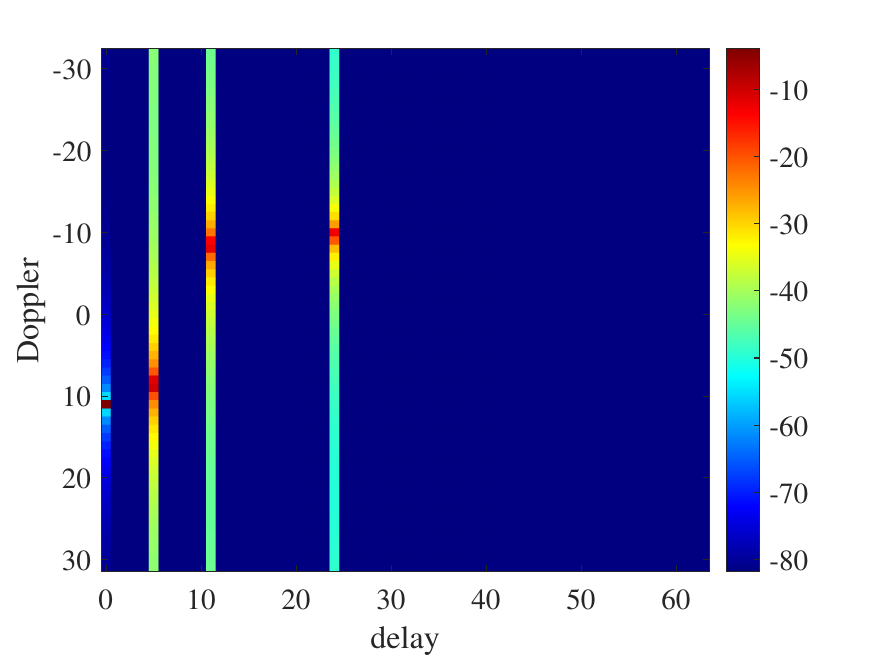}}
\subfloat[]{\label{subfig:DSE}
		\includegraphics[scale=0.4]{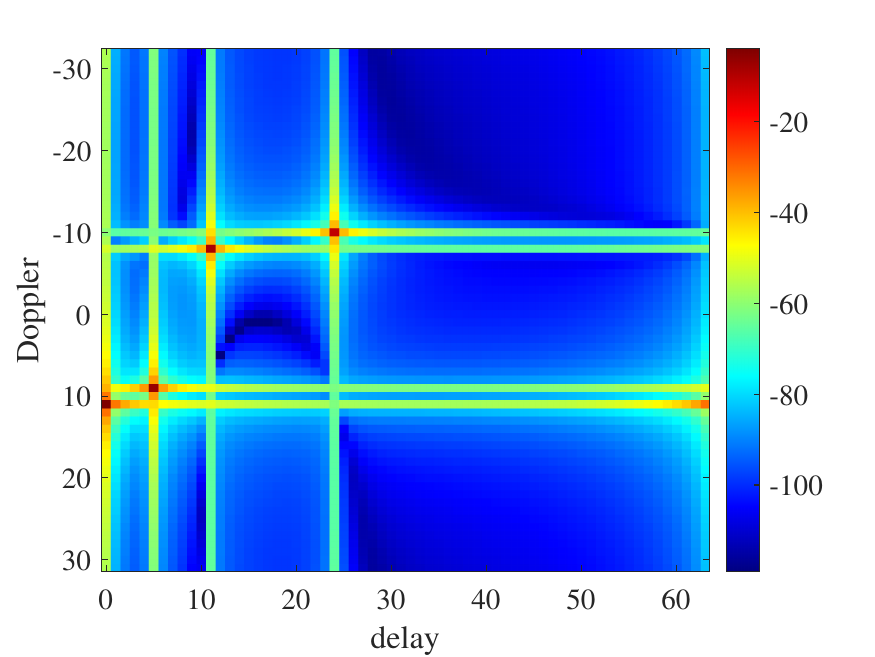}}
\subfloat[]{\label{subfig:frac_DSE}
		\includegraphics[scale=0.4]{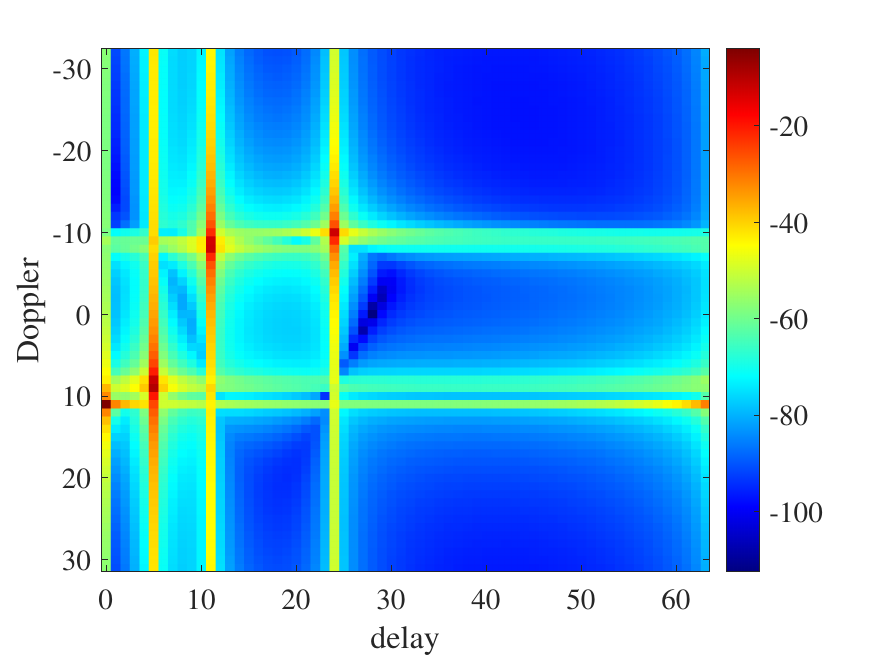}}
\caption{DD-domain channel amplitude (dB) under different conditions with $P=4$. (a) Presence of fractional Doppler. (b) Presence of DSE. (c) Presence of both fractional Doppler and DSE.}
\label{fig:H_DD_distributions}
\end{figure*}
The $(kM + l)$-th entry of the gradient vector
$\boldsymbol{\phi}_{p,\nu}'(k_i, l_i)$
is given by~\eqref{eq:derivative_k} at the bottom of next page,
\begin{figure*}[hb]
\centering
\hrulefill
\fontsize{9}{11}\selectfont 
\begin{equation} 
\begin{aligned}
&\boldsymbol{\phi}'_{p,\nu}(k_i, l_i) = \frac{1}{NM} \sum_{k^{\prime }= -k_{\max}}^{k_{\max}} \sum_{l^{\prime }=0}^{M-1}
j \pi \psi_i \left[{k^{\prime}, l^{\prime}}\right] \bigg( { a + \frac{2l_i}{N(M+M_\mathrm{cp})} } + \left(\frac{1}{N}+\frac{\Delta f (M-1)}{2f_{c}N} \right)  \mathcal{F}(b_i, N) + \frac{\Delta f (N-1)}{2f_{c}N}\mathcal{F}(c_i, M) \bigg) \\
&\qquad\qquad \times \theta_{k^{\prime}}(l,l^{\prime}) X^{\mathrm {DD}}\left [{\left \langle{ k-k^{\prime }}\right \rangle _{N}, \left ({l-l^{\prime}}\right )_{M}}\right].
\label{eq:derivative_k}
\end{aligned}
\end{equation}
\end{figure*}
where 
\begin{align*}
a &= \frac{2M_\mathrm{cp}}{N(M+M_\mathrm{cp})}
      + \frac{N-1}{N}
      + (N-1)(M-1)\frac{\Delta f}{2f_c N}, \\
b_i &= \frac{k_i \Delta f (M-1)}{2f_c N}
      - \frac{k' - k_i}{N}, \\
c_i &= \frac{k_i \Delta f (N-1)}{2f_c N}
      + \frac{l' - l_i}{M},
\end{align*}
and the function $\mathcal{F}(x_i, L)$ is defined as
\begin{equation}
\mathcal{F}(x_i, L) \triangleq
\frac{ \displaystyle\sum_{\ell=0}^{L-1} (2\ell - L + 1) e^{j\pi x_i (2\ell - L + 1)} }{ \displaystyle e^{-j\pi (L-1)x_i} \sum_{\ell=0}^{L-1} e^{j2\pi \ell x_i} }.
\end{equation}

Based on the virtual sampling grid, the satellite-terrestrial channel estimation problem in~\eqref{eq:hp_for_channel_estimation} can be reformulated as
\begin{equation}
\mathbf {y}_{p}=\overline{\boldsymbol{\Phi }}_{p}(\boldsymbol \kappa) \overline{\boldsymbol{\mathcal G}}  + \overline{\mathbf {v}}_{p},
\end{equation}
where $\overline{\boldsymbol{\Phi }}_{p} \in \mathbb{C}^{M_TN_T \times M_{\tau}N_{\nu}}$, $\overline{\boldsymbol{\mathcal G}}  \in \mathbb{C}^{M_{\tau}N_{\nu} \times 1}$, $\overline{\mathbf {v}}_{p} \in \mathbb{C}^{M_TN_T \times 1}$, and $\boldsymbol{\kappa} \in \mathbb{R}^{M_{\tau}N_{\nu} \times 1}$. The new measurement matrix $\overline{\boldsymbol{\Phi }}_{p}$ can be formulated as
\begin{align}
\overline{\boldsymbol{\Phi}}_p (\boldsymbol{\kappa}) = \boldsymbol{\Phi}_p + \boldsymbol{\Phi}_{p,\nu} \operatorname{diag}\big(\boldsymbol{\kappa} \big),
\end{align}
where 
\begin{align*}
\boldsymbol{\Phi}_p = &
\big[ \boldsymbol{\phi}_p(\overline k_0,0), \boldsymbol{\phi}_p(\overline k_1,0), \dots, \boldsymbol{\phi}_p(\overline k_{N_\nu-1},0),  \\&\boldsymbol{\phi}_p(\overline k_0,1), \boldsymbol{\phi}_p(\overline k_1,1),\dots, \boldsymbol{\phi}_p(\overline k_{N_\nu-1},1), \dots, \\&\boldsymbol{\phi}_p(\overline k_0,l_{\max}), \boldsymbol{\phi}_p(\overline k_1,l_{\max}),\dots, \boldsymbol{\phi}_p(\overline k_{N_\nu-1},l_{\max})\big],\\
\boldsymbol{\Phi}_{p,\nu}=&
\big[ \boldsymbol{\phi}_{p,\nu}'(\overline k_0,0), \boldsymbol{\phi}_{p,\nu}'(\overline k_1,0), \dots, \boldsymbol{\phi}_{p,\nu}'(\overline k_{N_\nu-1},0),  \\&\boldsymbol{\phi}_{p,\nu}'(\overline k_0,1), \boldsymbol{\phi}_{p,\nu}'(\overline k_1,1),\dots, \boldsymbol{\phi}_{p,\nu}'(\overline k_{N_\nu-1},1), \dots, \\&\boldsymbol{\phi}_{p,\nu}'(\overline k_0,l_{\max}), \boldsymbol{\phi}_{p,\nu}'(\overline k_1,l_{\max}),\dots, \boldsymbol{\phi}_{p,\nu}'(\overline k_{N_\nu-1},l_{\max})\big],
\end{align*} 
and $\boldsymbol{\kappa} = \big[ \kappa_0, \kappa_1, \dots, \kappa_{M_{\tau}N_{\nu}-1} \big]$.

\subsection{Prior-Aided Iterative Channel Reconstruction Scheme}
Fig.~\ref{fig:H_DD_distributions} illustrates the amplitude of the DD-domain channel under different channel conditions, highlighting the impact of fractional Doppler and DSE on energy dispersion.
As shown in Fig.~\ref{fig:H_DD_distributions}\subref{subfig:frac}, the presence of fractional Doppler causes the channel energy to spread from the dominant Doppler bin to its neighboring bins along the Doppler dimension. In Fig.~\ref{fig:H_DD_distributions}\subref{subfig:DSE}, DSE leads to additional energy dispersion in both the delay and Doppler dimensions. When fractional Doppler and DSE coexist, as illustrated in Fig.~\ref{fig:H_DD_distributions}\subref{subfig:frac_DSE}, the energy spreading effect is further aggravated, resulting in more severe inter-symbol interference (ISI). To mitigate ISI, we propose a PAICR algorithm, which consists of three main steps: channel coarse estimation, channel parameter extraction, and channel matrix recovery. The details of these steps are described in the following.

First, an initial channel estimate is obtained based on the pilot symbols and the received signals.
The DSE leads to energy leakage across multiple DD bins, thereby degrading the ideal sparsity of the channel. However, the channel energy remains concentrated in a limited region. SBL is therefore employed for coarse channel estimation, since it does not rely on strict sparsity assumptions and can provide robust estimates under approximately sparse conditions by automatically balancing sparsity promotion and noise suppression. The resulting estimate provides dominant path support and approximate parameter locations for the subsequent PAICR iterations, rather than being used as the final channel estimate.

Then, based on the channel parameters employed in the coarse estimation, i.e., 
$\boldsymbol{\hat{k}} = \overline{\boldsymbol{k}}_g + \boldsymbol{\kappa}$, 
$\boldsymbol{\hat{l}} = \overline{\boldsymbol{l}}_g$, and 
$\hat{\boldsymbol{\mathcal{G}}}$, we select the element 
$\hat{\mathcal{G}}_w$ of $\hat{\boldsymbol{\mathcal{G}}}$ with the largest magnitude, 
along with its corresponding Doppler and delay indices, denoted by 
$\hat{k}_w$ and $\hat{l}_w$. The resulting received signal is then expressed as
\begin{equation} 
\hat{\mathbf{y}}_p = \boldsymbol{\phi}_p(\hat{k}_w, \hat{l}_w) \hat{\mathcal{G}}_w,
\label{eq:y_hat} 
\end{equation}
where $w=1,2,\dots,W$ is the number of iterations. To ensure the accuracy of subsequent channel parameter extraction, the extracted dominant channel component should be eliminated, as given by
\begin{equation} 
\mathbf{y}_w = \mathbf{y}_{w-1} - \hat{\mathbf{y}}_p,
\label{eq:y_iteration_elimination} 
\end{equation}
where $\mathbf{y}_w$ denotes the received signal after eliminating the estimated dominant channel component during the $w$-th iteration.

According to the arrangement of the designed pilot symbols in~\eqref{eq:pilot_arrangement} and the DD-domain received signals in~\eqref{eq:Y_DDP}, we have
\begin{equation} 
\langle{ k-k^{\prime }} \rangle_{N} = 0,\;\text{i.e,}\;k=k^{\prime }.
\end{equation} 
This implies that prior knowledge of the stronger channel coefficients can be derived by examining the energy distribution in the Doppler dimension of the received signal. Specifically, the Doppler-domain energy distribution of the received signal for channel estimation can be expressed as
\begin{equation}
\mathbf E_w(k) = \frac{1}{M}{ \| \mathbf{Y}_w(k) \|_2},
\label{eq:doppler_energy} 
\end{equation}
where $\mathbf{Y}_w \in \mathbb{C}^{N_T \times M_T}$ is the matrixization of $\mathbf{y}_w$, and $k \in [-k_{\max}, k_{\max}]$ denotes the $k$-th row of $\mathbf{Y}_w$. Based on the Doppler-domain energy distribution of the received signal, the following criterion is defined to determine the convergence of the iterative process:
\begin{equation}
\epsilon_w = \left| \max_k \mathbf E_w(k) - \max_k \mathbf E_{w-1}(k) \right|.
\label{eq:convergence}
\end{equation}
The iteration is terminated when $\epsilon_w < 10^{-3}$. The satellite-terrestrial channel is then reconstructed based on the extracted channel parameters as
\begin{equation} 
\begin{aligned}
\label{eq:H_DD_hat}
&\hat H^{\mathrm{DD}}\left[{k^{\prime}, l^{\prime}}\right] = \sum_{w=1}^{W} \mathcal G_{w}
e^{j2\pi\frac{k_w\left(l_w+M_{\mathrm{cp}}\right)}{N\left(M+M_{\mathrm{cp}}\right)}}
e^{j\pi\left(N-1\right)\left(M-1\right)\frac{M+M_{\mathrm{cp}}}{2M\eta_{w}}}
\\&\hspace {2pc}
\times e^{j\pi\left(N-1\right)\frac{k_w-k^{\prime}}{N}}
\frac{\sin\pi N\left(\frac{k_{w}-k^{\prime}}{N}+(M-1)\frac{M+M_{\mathrm{cp}}}{2M\eta_{w}}\right)}{\sin\pi\left(\frac{k_{w}-k^{\prime}}{N}+(M-1)\frac{M+M_{\mathrm{cp}}}{2M\eta_{w}}\right)}
\\&\hspace {2pc}
\times e^{j\pi\left(M-1\right)\frac{l^{\prime}-l_{w}}{M}} \frac{\sin \pi M\left(\frac{l^{\prime}-l_{w}}{M}+(N-1)\frac{M+M_{\mathrm{cp}}}{2M\eta_{w}}\right)}{\sin\pi\left(\frac{l^{\prime}-l_{w}}{M}+(N-1)\frac{M+M_{\mathrm{cp}}}{2M\eta_{w}}\right)}.
\end{aligned}
\end{equation} 

\begin{figure*}[t]
\centering
\subfloat[]{\label{subfig:0iter}
		\includegraphics[scale=0.3]{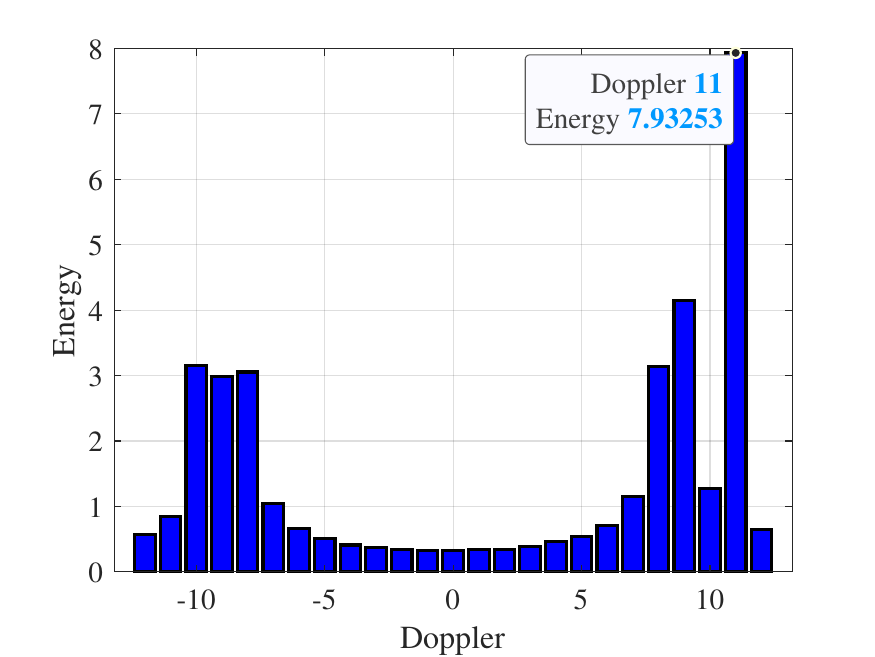}}
\subfloat[]{\label{subfig:1iter}
		\includegraphics[scale=0.3]{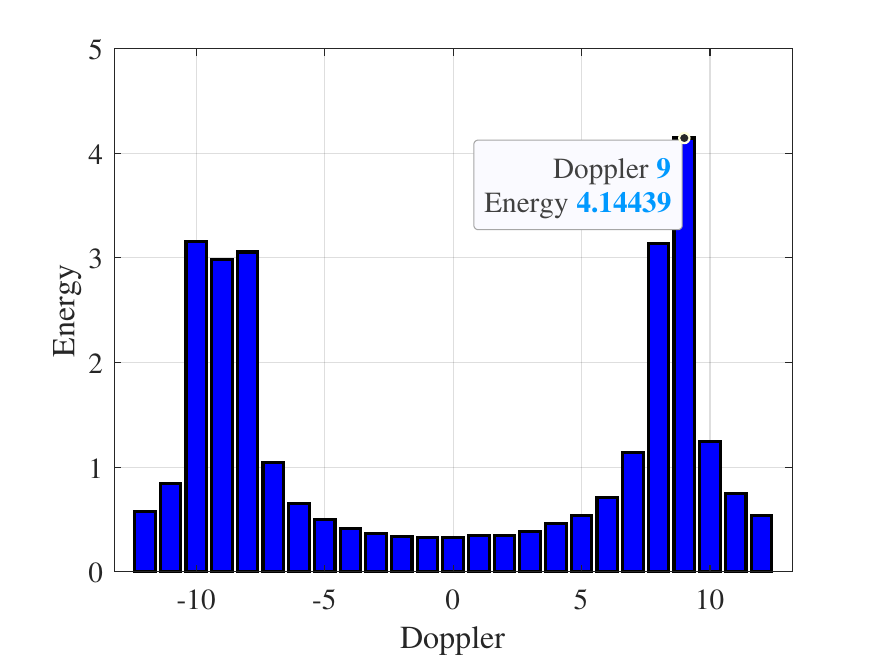}}
\subfloat[]{\label{subfig:2iter}
		\includegraphics[scale=0.3]{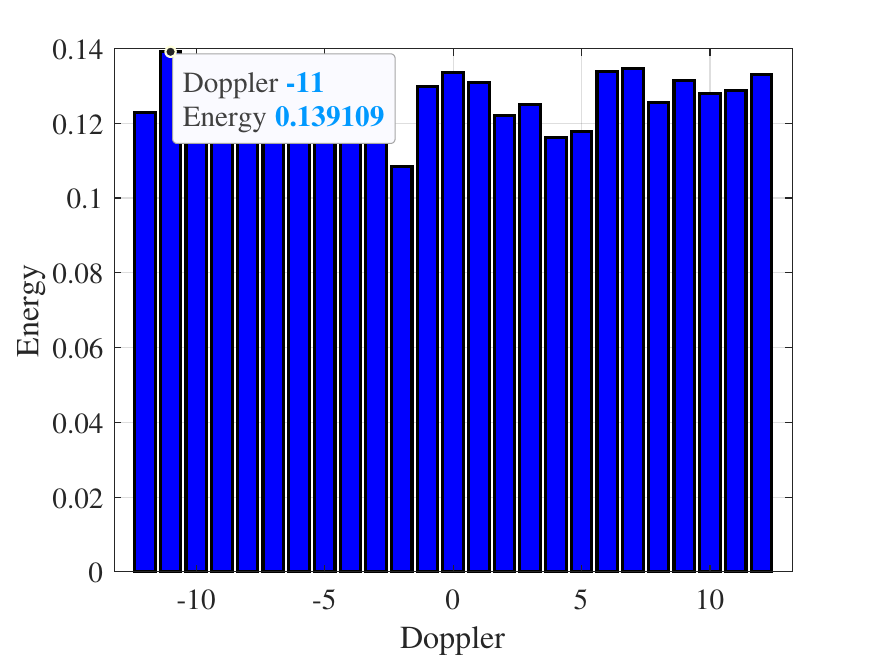}}
\subfloat[]{\label{subfig:3iter}
		\includegraphics[scale=0.3]{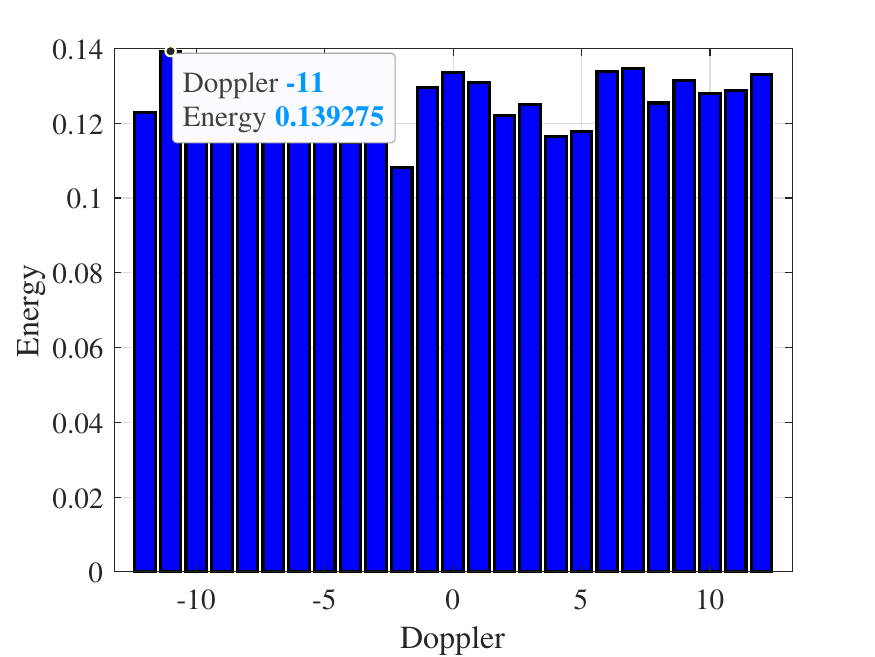}}
\caption{Doppler-dimension energy distribution of the received signal used for channel estimation during the iterative process ($\mathrm{SNR}_{\mathrm d}$ = 10 dB). (a) Original energy distribution. (b) Energy distribution after the 1st iteration. (c) Energy distribution after the $(W-1)$-th iteration. (d) Energy distribution after the $W$-th iteration.}
\label{fig:y_iter}
\end{figure*}

Fig.~\ref{fig:y_iter} illustrates the iterative elimination process of the received signal after passing through the channel shown in Fig.~\ref{fig:H_DD_distributions}\subref{subfig:frac_DSE}. 
Fig.~\ref{fig:y_iter}\subref{subfig:0iter} shows the energy distribution of the original received signal along the Doppler dimension. Fig.~\ref{fig:y_iter}\subref{subfig:1iter} presents the energy distribution after the first iteration, where the influence of the dominant channel component has already been removed. Fig.~\ref{fig:y_iter}\subref{subfig:2iter} depicts the state after the $(W-1)$-th iteration, when all dominant components have been eliminated, leaving only noise in the received signal. Finally, Fig.~\ref{fig:y_iter}\subref{subfig:3iter} shows that the maximum energy detected in the $W$-th iteration differs only slightly from that in the $(W-1)$-th iteration, indicating that all significant channel parameters have been successfully extracted and the iteration has reached its termination criterion. A summary of the proposed PAICR algorithm is provided in Algorithm~\ref{alg:PAICR}.

\begin{algorithm}[t]
\caption{Proposed PAICR Algorithm.}
\label{alg:PAICR}
\renewcommand{\algorithmicrequire}{\textbf{Input:}}
\renewcommand{\algorithmicensure}{\textbf{Output:}}
\begin{algorithmic}[1]
\REQUIRE $\mathbf {y}_{p}$, $\overline{\boldsymbol{\Phi }}_{p}$.
\ENSURE $\hat {\mathbf H}^{\mathrm {DD}}$.
\STATE Initialize $w=0$, $\mathbf{y}_0 = \mathbf{y}_p$.
\STATE \textbf{repeat} 
\STATE \hspace{0.5cm} Set $w = w+1$.
\STATE \hspace{0.5cm} Coarse channel estimation using SBL algorithm.
\STATE \hspace{0.5cm} Extract channel parameters $\hat {k}_w, \hat {l}_w, \hat{\mathcal{G}}_w$.
\STATE \hspace{0.5cm} Compute $\hat{\mathbf{y}}_p$ by~\eqref{eq:y_hat}.
\STATE \hspace{0.5cm} Update $\mathbf{y}_w$ by~\eqref{eq:y_iteration_elimination}.
\STATE \hspace{0.5cm} Compute $\mathbf E_w$ by~\eqref{eq:doppler_energy}.
\STATE \hspace{0.5cm} Compute $\epsilon_w$ by~\eqref{eq:convergence}.
\STATE \textbf{until} $\epsilon_w < 10^{-3}$. 
\STATE Reconstruct the channel $\hat {\mathbf H}^{\mathrm {DD}}$ by~\eqref{eq:H_DD_hat}.
\end{algorithmic}
\end{algorithm}


\subsection{Complexity Analysis}

\begin{table}[t]
\begingroup
\caption{Comparison of Complexity Orders\label{tab:complexity}}
\centering
\renewcommand{\arraystretch}{1.8} 
\begin{tabular}{>{\centering\arraybackslash}p{2cm}|>{\centering\arraybackslash}p{6cm}} 
\hline 
\textbf{Scheme} & \textbf{Complexity order} \\
\hline
AMP~\cite{8836636}  & $\mathcal{O}\left( N_{\text{iter}} \left( M_TN_T \right)^2 \right)$ \\
\hline
VAMP~\cite{8713501}  & $\mathcal{O}\left( M_TN_T \min \left( M_T,N_T\right) + N_{\text{iter}} R M_TN_T \right)$ \\
\hline
OMP~\cite{10103827}  & $\mathcal{O}\left( N_{\text{iter}}^3 M_TN_T M_{\tau}N_{\nu} \right)$ \\
\hline
SBL~\cite{10475894}  & $\mathcal{O}\left( N_{\text{iter}} M_TN_T \left(M_{\tau}N_{\nu}\right)^2 \right)$ \\
\hline
Proposed PAICR  & $\mathcal{O}\left( W \left(N_{\text{iter}} M_TN_T \left(M_{\tau}N_{\nu}\right)^2 + \left(M_TN_T\right)^3 \right)\right)$ \\
\hline
\end{tabular}
\endgroup
\end{table}

In this subsection, the complexity order of the proposed scheme is analyzed and compared with several benchmark algorithms, as summarized in Table~\ref{tab:complexity}. Here, $R$ denotes the rank of the measurement matrix and $N_{\text{iter}}$ denotes the maximum iteration number of the considered algorithms.
The overall complexity of the proposed PAICR algorithm mainly comprises two components: coarse channel estimation and iterative channel parameter extraction. The complexity of the coarse channel estimation is determined by the SBL algorithm, with complexity order $\mathcal{O}\left( N_{\text{iter}} M_TN_T \left(M_{\tau}N_{\nu}\right)^2 \right)$. 
In the iterative channel parameter extraction, the major computational burden lies in the computation of $\hat{\mathbf{y}}_p$ at each iteration, whose complexity order is $\left(M_TN_T\right)^3$. 
Therefore, the total complexity order of the proposed PAICR algorithm is given by $\mathcal{O}\left( W \left(N_{\text{iter}} M_TN_T \left(M_{\tau}N_{\nu}\right)^2 + \left(M_TN_T\right)^3 \right)\right)$.
Despite its relatively high complexity order, the proposed scheme achieves a favorable performance-complexity trade-off.
As demonstrated in the simulation results shown in Fig.~\ref{fig:numIter} in Section~\ref{section:5}, the number of outer iterations $W$ typically remains below 12.
More importantly, the increased computational cost is mainly incurred to enhance the accuracy of channel estimation, by iteratively extracting channel parameters to effectively suppress the severe ISI induced by DSE.
Therefore, given the limited number of outer iterations and the substantial improvement in channel estimation performance, the overall complexity remains acceptable for practical implementation.


\subsection{Cramér-Rao Lower Bound}
To characterize the fundamental performance limits of CP-OTFS-based LEO satellite systems in terms of channel estimation, the CRLB is derived in this subsection.

According to \eqref{eq:hp_for_channel_estimation}, the unknown parameters are $\boldsymbol{\mathcal G} = \left [ \mathcal G_{1}, \mathcal G_{2}, \dots , \mathcal G_{P} \right ]$, $\boldsymbol{k} = \left[ k_1, k_2, \dots, k_{P}\right]$ and $\boldsymbol{l} = \left[ l_1, l_2, \dots, l_{P}\right]$, therefore we aim at deriving the Fisher information matrix (FIM). For the $i$-th path, the partial derivatives of $\mathbf{H}^{\mathrm{DD}}$ with respect to the channel gain $\mathcal G_i$, Doppler index $k_i$, and delay index $l_i$ are respectively given by~\eqref{eq:derivative_crlb} at the bottom of this page.
\begin{figure*}[hb]
\centering
\hrulefill
\fontsize{9}{11}\selectfont
\begin{equation}
\left\{
\begin{aligned}
\frac{\partial {H}^{\mathrm{DD}}[k',l']}{\partial {\mathcal G_{i}}}
&= e^{-j \pi \frac{ k'(N-1)}{N}} e^{j 2\pi \frac{k_{i} l_{i}}{N(M+M_\mathrm{cp})}}
e^{j \pi \frac{ (l'-l_i) (M-1)}{M}}  e^{j \pi k_{i} a} e^{-j\pi (N-1)b_i} e^{-j\pi (M-1)c_i} \sum_{n=0}^{N-1} e^{j2\pi n b_i}  \sum_{m=0}^{M-1} e^{j2\pi m c_i}, \\[1ex]
\frac{\partial {H}^{\mathrm{DD}}[k',l']}{\partial k_{i}}
&= j \pi{H}^{\mathrm{DD}}[k',l']
\bigg( {  a + \frac{2l_i}{N(M+M_\mathrm{cp})} } + \left(\frac{1}{N}+\frac{\Delta f (M-1)}{2f_{c}N} \right)
\mathcal{F}(b_i, N)
+ \frac{\Delta f (N-1)}{2f_{c}N}
\mathcal{F}(c_i, M) \bigg), \\[1ex]
\frac{\partial {H}^{\mathrm{DD}}[k',l']}{\partial l_{i}}
&= j\pi{H}^{\mathrm{DD}}[k',l']
\bigg( \left( \frac{2k_i}{N(M+M_\mathrm{cp})}-\frac{M-1}{M}\right)
-\frac{1}{M}  \mathcal{F}(c_i, M)\bigg).
\label{eq:derivative_crlb}
\end{aligned}
\right.
\end{equation}
\end{figure*}
We define the gradient matrix as
\begin{equation} 
\mathbf {G} = [\boldsymbol{\gamma}_1, \boldsymbol{\gamma}_2, \dots , \boldsymbol{\gamma}_{P}],
\end{equation}
where $\boldsymbol{\gamma}_i = [\text{vec}(\frac{\partial \mathbf {H}^{\mathrm{DD}}}{\partial {\mathcal G_{i}}}), \text{vec}(\frac{\partial \mathbf {H}^{\mathrm{DD}}}{\partial {k_{i}}}), \text{vec}(\frac{\partial \mathbf {H}^{\mathrm{DD}}}{\partial {l_{i}}})]$. Then, the FIM is given by
\begin{equation} 
\mathbf{J} = \frac{2}{\sigma ^2} \Re \left\{ \left(  \mathbf{ \widetilde \Phi }_p \mathbf{G} \right)^H \left( \mathbf{ \widetilde \Phi }_p \mathbf{G} \right) \right\},
\end{equation}
where $\boldsymbol{\widetilde \Phi }_{p} \in \mathbb{C}^{M_TN_T \times M_TN_T}$ is the sensing matrix without channel information for pilot symbols, which can be expressed as
\begin{equation} 
\begin{aligned}
\boldsymbol{\widetilde \Phi }_{p}[Mk'+l',Mk+l]&=\frac{1}{NM}\theta_{k^{\prime}}(l,l^{\prime}) \\
&\times X^{\mathrm {DD}}\left [{\left \langle{ k-k^{\prime }}\right \rangle _{N}, \left ({l-l^{\prime}}\right )_{M}}\right].
\end{aligned}
\end{equation}
Consequently, the CRLB of $\mathbf {H}^{\mathrm{DD}}$ is obtained as
\begin{equation} 
\text{CRLB}(\mathbf {H}^{\mathrm{DD}}) = \mathrm{tr}\left( \mathbf{G} \mathbf{J}^{-1} \mathbf{G}^{H} \right),
\end{equation}
which provides a theoretical lower bound on the estimation error variance of $\mathbf {H}^{\mathrm{DD}}$ and serves as a benchmark for evaluating the performance loss of practical channel estimation algorithms.
\section{Numerical Results}\label{section:5}

\begin{table}[t]
\caption{Simulation Parameters\label{tab:parameter}}
\centering
\begin{tabular}{>{\centering\arraybackslash}p{4.5cm}|>{\centering\arraybackslash}p{2.5cm}} 
\hline
Parameter & Values
\\
\hline
Earth radius  & 6371 km \\
Satellite height  & 1000 km \\
Elevation angle & $50^{\circ}$ \\
Satellite speed & 7562.2 m/s \\
Terminal speed & 50 m/s \\
Modulation alphabet & 4-QAM \\
Carrier frequency $(f_c)$ & 2 GHz \\
Subcarrier spacing $(\Delta f)$ & 240 kHz \\
Number of subcarriers $(M)$ & [32, 64, 128] \\
Number of time slots $(N)$ & [32, 64, 128] \\
Length of CP $(M_{\mathrm{cp}})$ & $M/4$ \\
\hline
\end{tabular}
\end{table}

In this section, we conduct numerical results to evaluate the performance of the proposed transmission scheme. 
We employ the NTN-TDL-B channel model~\cite{3GPP}, where the number of paths is set to $P=4$, and each delay tap exhibits a single Doppler shift generated according to the Jakes' model.
Table~\ref{tab:parameter} summarizes the typical values of relevant simulation parameters.
The signal-to-noise ratios (SNRs) for data and pilot symbols are defined as
$\mathrm{SNR}_{\mathrm d} = \mathbb{E}\{|X^{\mathrm{DD}}_{\mathrm d}|^2\}/\sigma^2$
and
$\mathrm{SNR}_{\mathrm p} = \mathbb{E}\{|X^{\mathrm{DD}}_{\mathrm p}|^2\}/\sigma^2$,
respectively, where $\sigma^2$ denotes the noise variance.
Unless otherwise specified, all simulations are conducted with a frame size of $N = 32$ and $M = 32$, and the transmit power of pilot symbols is set to be 30~dB higher than that of data symbols. Simulation results are averaged over 600 randomized trials.

\subsection{Performance Evaluation of the OTFS Frame Structure}
To evaluate the PAPR and OOBE performance of the proposed OTFS frame structure, we select the embedded pilot (EP) scheme~\cite{8671740} and the superimposed pilot (SP) scheme~\cite{9539066} as benchmark schemes. In both the EP and the proposed schemes, the power of the pilot symbols is 30 dB higher than that of the data symbols, while in the SP scheme, the pilot symbol power is set to $\frac{3}{7}$ of the data symbol power. Moreover, all signals are oversampled by a factor of 4 like~\cite{tellado2005multicarrier}. To ensure fairness, all generated time-domain signals are normalized in power.

\begin{figure}[!t]
\centering
\subfloat[]{
\includegraphics[scale=0.6]{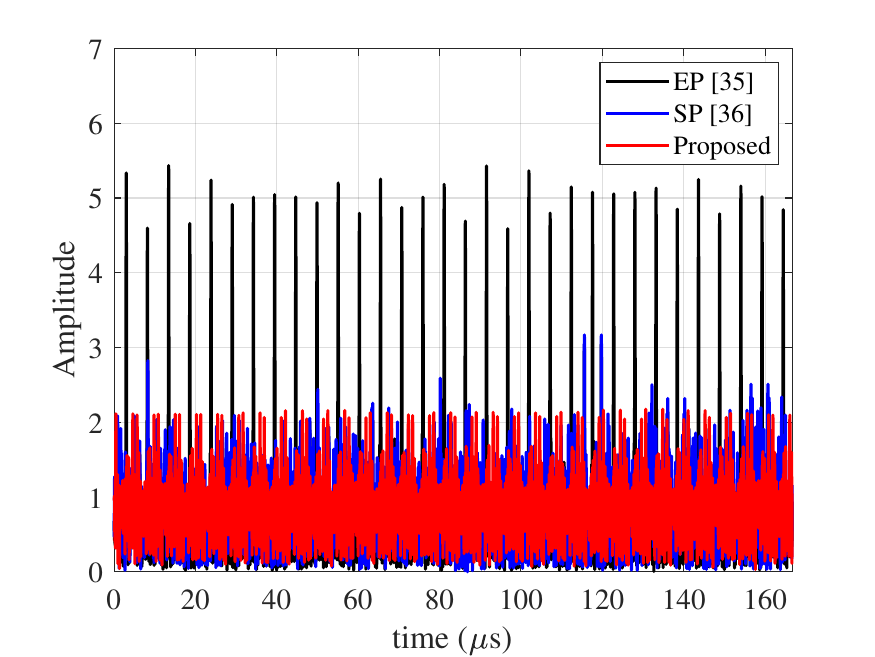}
\label{fig:amplitude}}
\hfil
\subfloat[]{
\includegraphics[scale=0.6]{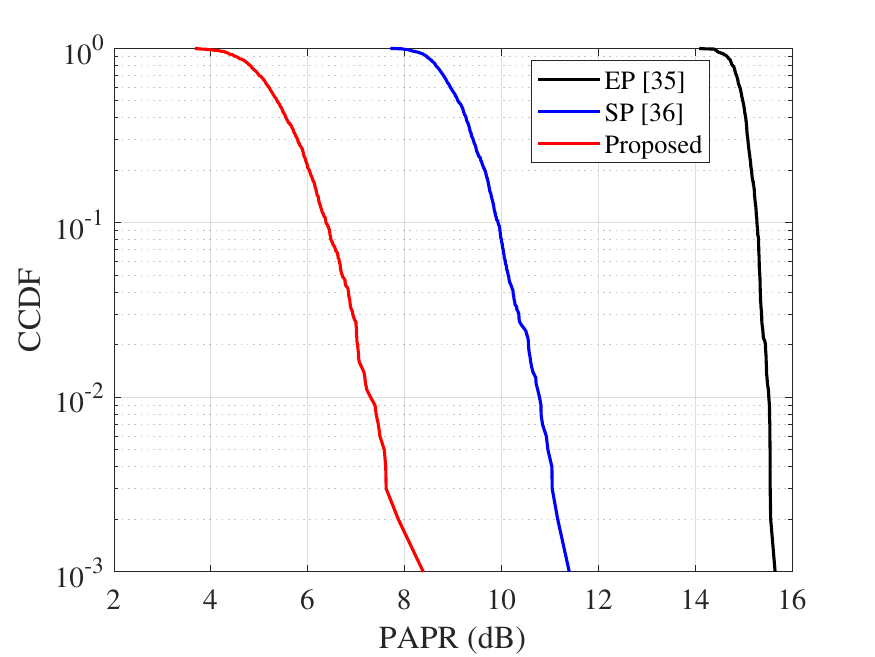}
\label{fig:ccdf}}
\caption{PAPR performance comparison of different pilot arrangement schemes. (a) Time-domain amplitude distributions. (b) CCDF of the PAPR.}
\label{fig:PAPR}
\end{figure}

Fig.~\ref{fig:PAPR} compares the time-domain amplitude distributions and PAPR performance of the EP, SP, and the proposed scheme.
Fig.~\ref{fig:PAPR}\subref{fig:amplitude} shows the time-domain amplitude distributions for the three schemes. Most samples across all schemes exhibit amplitudes below 2. However, the EP scheme, due to the isolated insertion of a high-power pilot symbol in the DD domain, produces significant peaks in the time domain, with a maximum amplitude exceeding 5. The SP scheme, which superimposes a low-power pilot on the original data symbol, mitigates the peaks but still shows abrupt variations. In contrast, the proposed pilot scheme, despite using high-power pilots, concentrates all pilots at $k=0$. After mapping to the time domain, this yields a smoother waveform without sharp peaks, maintaining peak amplitudes below 2.2.
Fig.~\ref{fig:PAPR}\subref{fig:ccdf} presents the complementary cumulative distribution functions (CCDF) of the PAPR for the three schemes. The results confirm that the proposed pilot scheme achieves the lowest PAPR levels. This is primarily because, after ISFFT and OFDM modulation, the pilot energy in the proposed scheme spreads more evenly across the entire time-domain waveform, increasing the average power and balancing the power distribution of the signal, thereby significantly reducing the PAPR.

\begin{figure}[ht]
\centering
\includegraphics[scale=0.6]{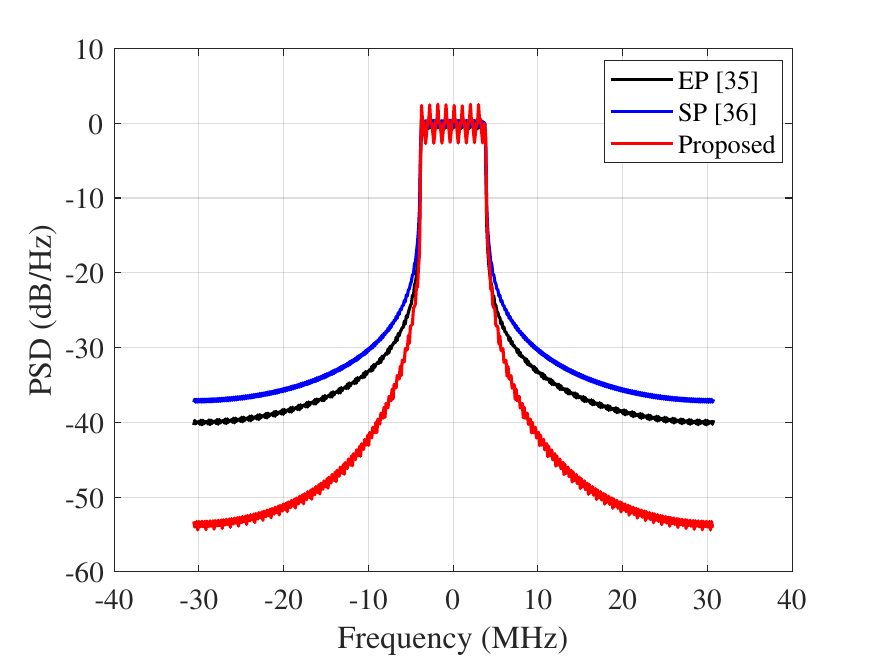}
\caption{OOBE performance comparison of different pilot arrangement schemes.}
\label{OOBE}
\end{figure}

Fig.~\ref{OOBE} compares the performance of the EP, SP, and the proposed frame structure design in terms of the power spectral density (PSD). It should be noted that the spectral leakage in PSD is a direct indicator of OOBE level. 
Simulation results show that the proposed pilot scheme achieves the best performance in OOBE suppression. This improvement is primarily attributed to the continuity of the time-domain waveform at symbol stitching boundaries (i.e., at $t=nT_\mathrm{sym}$), which effectively reduces waveform discontinuities and thereby suppresses spectral leakage. Furthermore, the proposed design retains the simple implementation of the rectangular pulse structure while optimizing spectral characteristics, demonstrating a good balance between practicality and performance.

\subsection{Channel Estimation Performance}
In this subsection, we evaluate the proposed PAICR algorithm against several benchmark methods, including AMP~\cite{8836636}, VAMP~\cite{8713501}, OMP~\cite{10103827}, and SBL~\cite{10475894} algorithms.
The normalized mean square error (NMSE) is defined as ${\mathrm{ NMSE}} = \frac {\left \|{\bf {h}- \bf {\hat {h}}}\right \|_{2}^{2}}{\left \|{\mathbf{h}}\right \|_{2}^{2}}$, where $\mathbf{h}$ and ${\hat {\mathbf{h}}}$ denote the vectorization of the true and estimated DD-domain channel matrices $\mathbf H^{\mathrm {DD}}$ and $\hat {\mathbf H}^{\mathrm {DD}}$, respectively.

\begin{figure}[t]
\centering
\includegraphics[scale=0.6]{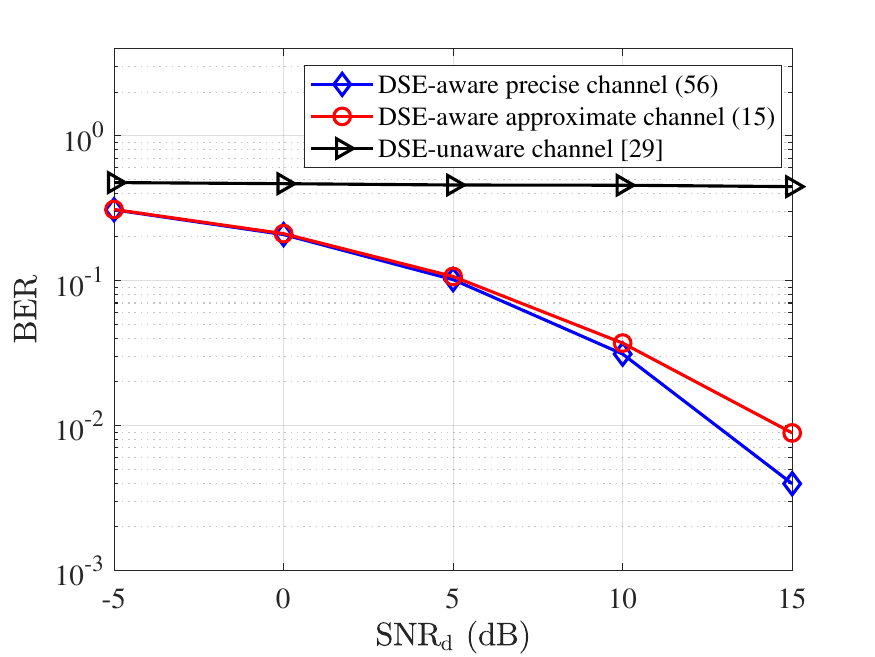}
\caption{BER performance comparison under MMSE-based data detection for different channel characterizations.}
\label{fig:ber}
\end{figure}

Fig.~\ref{fig:ber} shows the BER performance under MMSE-based data detection for distinct channel characterizations, including the DSE-aware precise channel in~\eqref{eq:preH}, the DSE-aware approximate channel in~\eqref{eq:H_DD}, and the DSE-unaware channel~\cite{10679987}. It can be observed that the DSE-unaware channel yields the worst detection performance, as it fails to account for the power dispersion effect in the channel, leading to modeling inaccuracies. In contrast, the DSE-aware approximate channel performs similarly to the DSE-aware precise channel under low SNR conditions. When SNR = 15 dB, the BER of the DSE-aware approximate channel is only 0.0049 higher than that of the precise model. This result further verifies the effectiveness and practicality of the proposed approximation model.

\begin{figure}[t]
\centering 
\includegraphics[scale=0.6]{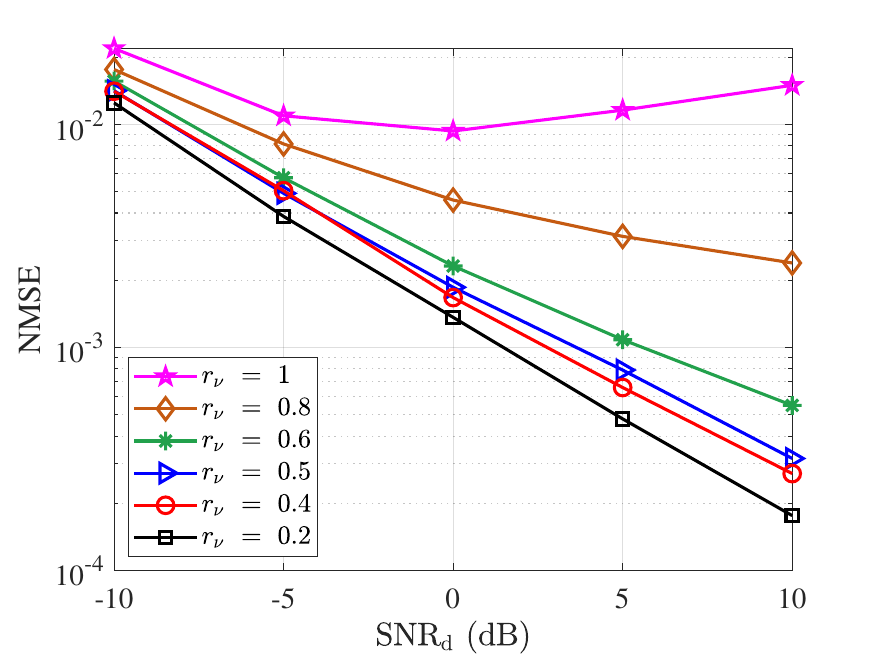}
\caption{NMSE performance comparison of the proposed PAICR algorithm for different virtual Doppler resolutions.} 
\label{fig:resolution}
\end{figure}

Fig.~\ref{fig:resolution} illustrates the NMSE performance of the proposed PAICR algorithm under different virtual Doppler resolutions. It is observed that decreasing $r_\nu$ from 1 to 0.6 leads to a pronounced NMSE reduction, indicating that finer virtual Doppler grids help better capture the dominant channel components. However, when $r_\nu$ is further reduced from 0.5 to 0.2, only marginal performance improvement is achieved, with the NMSE decreasing by approximately $0.0003$ at $\mathrm{SNR}_\mathrm{d} = 10$~dB. This behavior suggests that the proposed PAICR algorithm is insensitive to excessively fine virtual Doppler resolutions and can already achieve near-saturated performance with a moderate grid density. Considering that a smaller $r_\nu$ inevitably incurs higher computational complexity, $r_\nu = 0.5$ is therefore selected as the default virtual Doppler resolution for subsequent simulations, striking an effective balance between estimation accuracy and computational efficiency.

\begin{figure}[!t]
\centering 
\subfloat[]{
\includegraphics[scale=0.6]{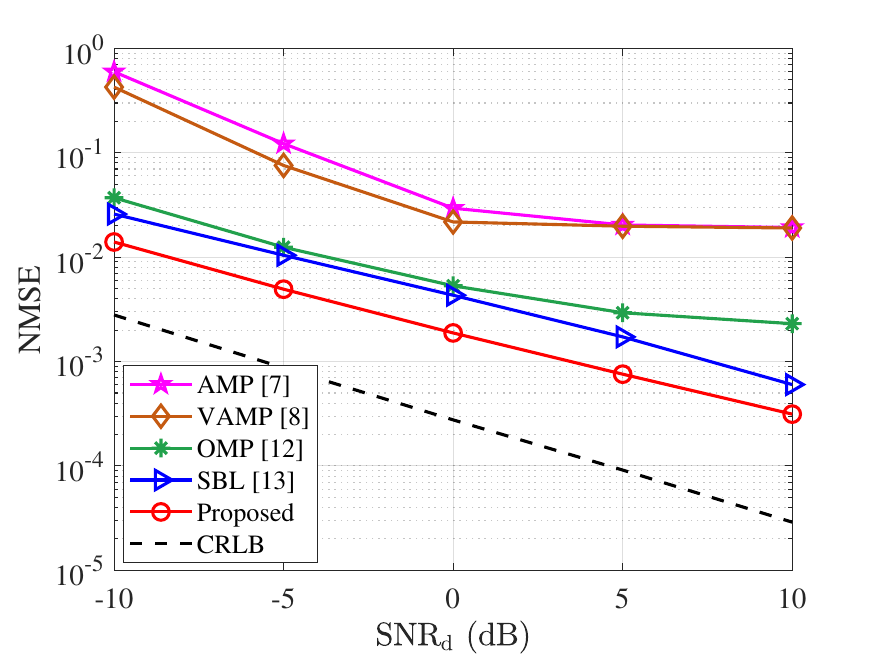}
\label{fig:M32N32_d}}
\hfil
\subfloat[]{
\includegraphics[scale=0.6]{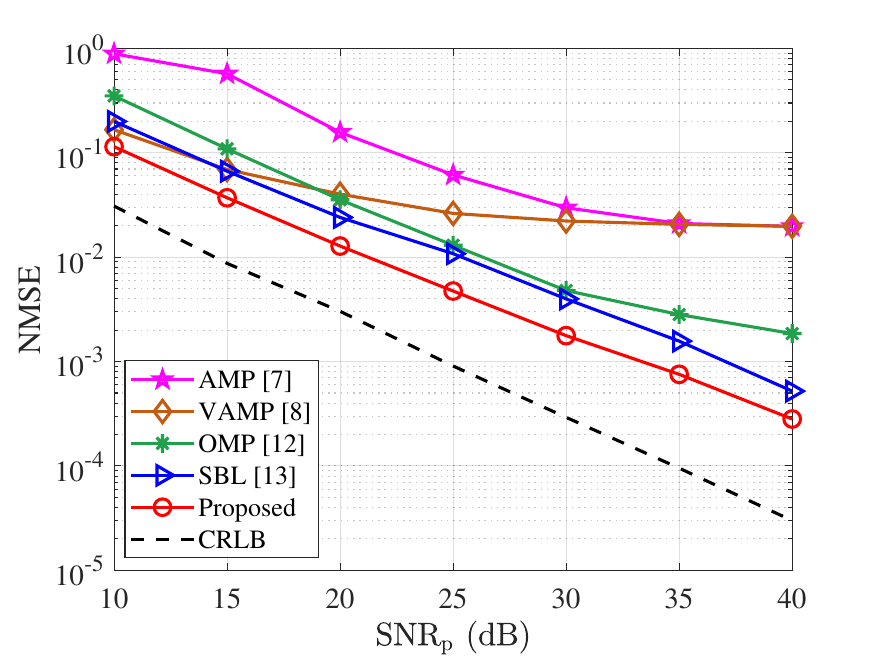}
\label{fig:snr_p}}
\caption{NMSE performance comparison of different algorithms.  (a) NMSE versus $\mathrm{SNR}_\mathrm{d}$. (b) NMSE versus $\mathrm{SNR}_\mathrm{p}$ with $\mathrm{SNR}_\mathrm{d}=10$~dB.} 
\label{fig:M32N32}
\end{figure}

Fig.~\ref{fig:M32N32} compares the NMSE performance of different channel estimation algorithms.
As shown in Fig.~\ref{fig:M32N32}\subref{fig:M32N32_d}, when $\mathrm{SNR}_\mathrm{d}$ increases from $-10$ dB to 10 dB, the proposed PAICR algorithm consistently achieves the lowest NMSE among all benchmark schemes. This performance gain arises because conventional AMP- and VAMP-based algorithms fail to account for the channel power leakage induced by DSE. In contrast, the proposed PAICR algorithm starts from an initial coarse channel estimate obtained via the SBL algorithm and further iteratively extracts dominant channel parameters while eliminating their contributions from the received signal. By progressively mitigating the interference caused by power leakage, the proposed approach is able to reconstruct a more complete and accurate channel response.
Fig.~\ref{fig:M32N32}\subref{fig:snr_p} further illustrates the NMSE performance versus $\mathrm{SNR}_\mathrm{p}$ with $\mathrm{SNR}_\mathrm{d}$ fixed at 10 dB. It can be observed that increasing the pilot power improves the estimation accuracy of all algorithms. However, the proposed PAICR algorithm consistently maintains a clear performance advantage over the benchmark schemes and converges closer to the CRLB.
Overall, these results confirm that by leveraging prior information from Doppler-domain energy observations and iterative dominant component cancellation, the proposed PAICR algorithm enables accurate channel reconstruction under DSE-aware channel conditions, outperforming existing benchmark methods.


\begin{figure}[!t]
\centering
\subfloat[]{
\includegraphics[scale=0.6]{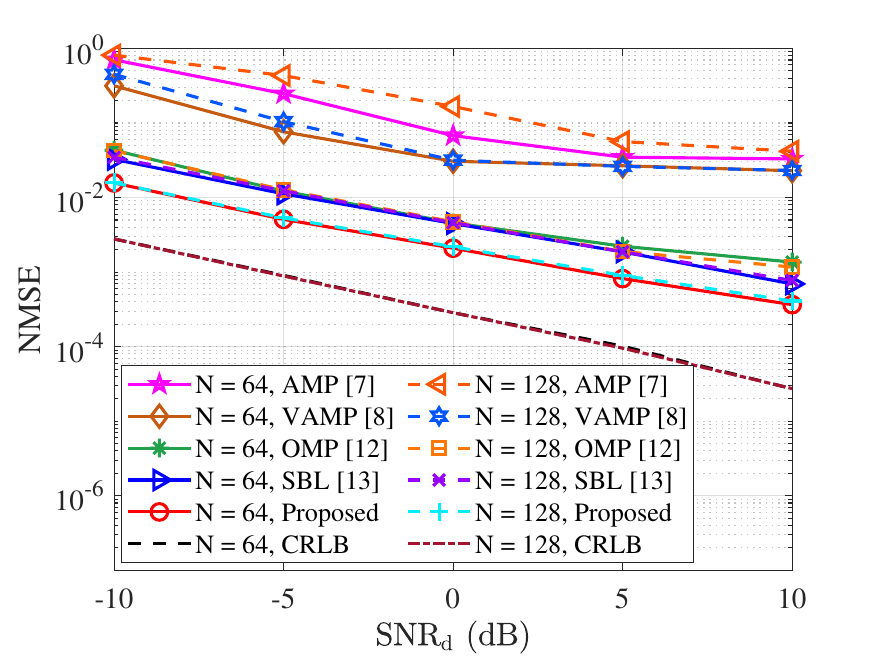}
\label{fig:N}}
\hfil
\subfloat[]{
\includegraphics[scale=0.6]{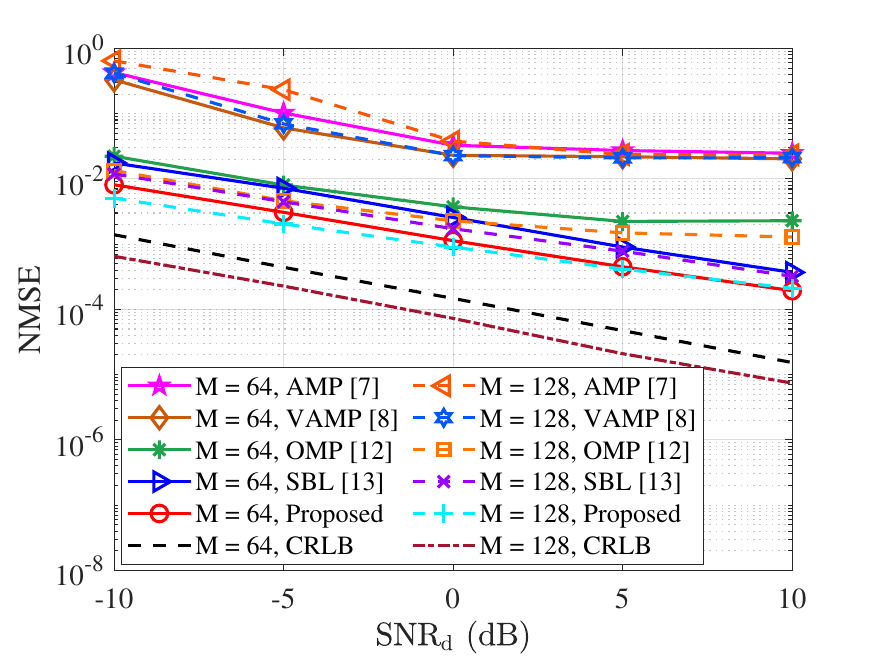}%
\label{fig:M}}
\caption{NMSE performance vs. frame size. (a) Fixed $M = 32$ with varying $N$. (b) Fixed $N = 32$ with varying $M$.}
\label{fig:frame}
\end{figure}

Fig.~\ref{fig:frame} examines the impact of different frame sizes on channel estimation performance in terms of NMSE. According to~\eqref{eq:H_DD}, as $M$ and $N$ increase, the main lobe of the sinc function becomes narrower and the number of sidelobes increases, causing the channel energy to spread over multiple DD-domain positions. This reduces the sparsity of the channel, thereby degrading estimation performance.
Fig.~\ref{fig:frame}\subref{fig:N} presents the estimation performance under different values of $N$. 
It can be observed that the estimation performance of the PAICR remains stable as $N$ increases.
This is because the proposed PAICR algorithm reconstructs the channel based on extracted physical channel parameters, allowing it to effectively recover energy leakage caused by DSE.
Fig.~\ref{fig:frame}\subref{fig:M} shows the results for different values of $M$. Unlike the case with $N$, the estimation performance of the proposed PAICR algorithm improves as $M$ increases. This is attributed to the proposed OTFS frame structure, which distributes pilot symbols uniformly along the delay dimension, enabling effective coverage of the dispersed channel energy. As a result, the accuracy of parameter extraction is improved, leading to enhanced overall estimation performance.

\begin{figure}[t]
\centering 
\includegraphics[scale=0.6]{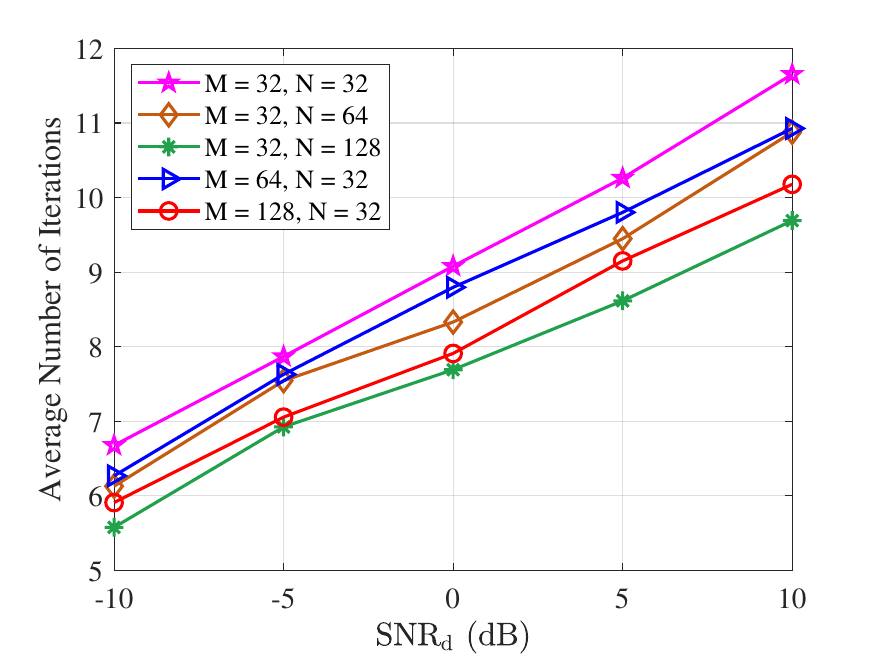}
\caption{Average number of iterations of the proposed PAICR algorithm versus frame size.} 
\label{fig:numIter}
\end{figure}

Fig.~\ref{fig:numIter} depicts the average number of outer iterations required by the proposed PAICR algorithm under different frame sizes. It can be observed that the average iteration number increases moderately with $\mathrm{SNR}_\mathrm{d}$, since a higher $\mathrm{SNR}_\mathrm{d}$ enables the extraction of more resolvable channel components, which in turn requires additional iterative refinements. 
Moreover, smaller frame sizes generally result in a slightly higher iteration count, as the reduced observation dimensionality leads to less concentrated Doppler-domain energy representations, such that each iteration can only extract and suppress a limited portion of the dominant channel components. Consequently, more outer iterations are required to progressively mitigate the residual power leakage.
Nevertheless, for different frame sizes, the number of outer iterations remains below 12, confirming that the proposed algorithm can converge rapidly.
More importantly, the additional computational burden is mainly devoted to iteratively extracting dominant channel parameters and suppressing the severe ISI induced by DSE. As a result, the overall computational complexity of the proposed PAICR algorithm remains acceptable.

\section{Conclusion}\label{section:6}
In this paper, we developed a DSE-resilient transmission scheme tailored for CP-OTFS-based LEO satellite communication systems, where a novel OTFS frame structure was optimized for both low OOBE and PAPR. We conducted DSE-aware channel characterization of the satellite-terrestrial channel in the DD domain, which unveiled that DSE caused power leakage, thus undermining reliable parameter estimation. 
To resolve this, we proposed a PAICR algorithm, which mitigated DSE-induced distortions and enabled accurate path parameter extraction.
A corresponding CRLB was derived to provide a theoretical performance benchmark.
Simulation results underscored the critical need to account for DSE in CP-OTFS-based LEO satellite systems and validated the effectiveness of the proposed approach.

{\appendix[Proof of Proposition 1]}
\allowdisplaybreaks
According to~\eqref{eq:Y_TF}, we have
\begin{equation} 
\begin{aligned}
&\hspace {-1pc}{{Y}^\mathrm{TF}}[n,m]\\
=&\sum_{m^{\prime}=0}^{M-1}\sum_{n^{\prime}=0}^{N-1}{{X}^\mathrm{TF}}\left [{ {n^{\prime},m^{\prime}} }\right]\iint{h\left(\tau ,\nu \right)}{{e}^{-j 2\pi m\Delta f \tau}}\\
&{{A}_{{{g}_{\mathrm{rx}}},{{g}_{\mathrm{tx}}}}}\left (\left (n-{n^{\prime}}\right ){{T}_{\mathrm{sym}}}-\tau ,\left(m-m^{\prime}\right)\Delta f-\nu \right ) \\
&{{e}^{j 2\pi \nu \left(\tau +{n^{\prime}}{{T}_{\mathrm{sym}}}+{{T}_{\mathrm{cp}}}\right)}}d\tau d\nu\\
\overset{\left(a\right)}{=}&\frac{1}{M}\sum_{m^{\prime}=0}^{M-1}\sum_{n=0}^{N-1}{{X}^\mathrm{TF}}\left [{ {n,m^{\prime}} }\right]\iint{h\left(\tau ,\nu \right)} {{e}^{-j 2\pi m\Delta f \tau}}\\
&\sum\limits_{q={{M}_{\mathrm{cp}}}-{{l}_{\tau }}}^{M+{{M}_{\mathrm{cp}}}-1-{{l}_{\tau }}}{{{e}^{-j 2\pi \left(\left(m-{m}^{\prime}\right)\Delta f-\nu \right)\frac{q}{M\Delta f}}}}{{e}^{j 2\pi \left(m-{m}^{\prime}\right)\Delta f{{T}_{\mathrm{cp}}}}}\\
&{{e}^{j 2\pi \nu \left(\tau +n{{T}_{\mathrm{sym}}}\right)}}d\tau d\nu \\ 
= & \frac{1}{\sqrt{MN}} \sum_{n=0}^{N-1} \sum _{k^{\prime}=\left\lceil -\frac{N}{2} \right\rceil}^{\left\lceil \frac{N}{2} \right\rceil - 1} \sum_{l^{\prime}=0}^{M-1}{{X}^{DD}}[{k}^{\prime},{l}^{\prime}]\iint{h\left(\tau ,\nu \right) } {{e}^{j2\pi \frac{m{{M}_{\mathrm{cp}}}}{M}}}\\
&{{e}^{j2\pi \frac{n{k}^{\prime}}{N}}} {e}^{j2\pi \nu n T_{\mathrm{sym}}}
{{e}^{j2\pi \left(\nu- m\Delta f\right) \left(\tau + \frac{{{M}_{\mathrm{cp}}}+l^{\prime}}{M\Delta f}\right)}} d\tau d\nu, 
\label{Y_TF1}
\end{aligned}
\end{equation} 
where $(a)$ is due to ${{l}_{\tau }}\triangleq \left\lceil \tau /\left(T/M\right) \right\rceil$, $T_{\mathrm{cp}}>\tau$ and ${{A}_{{{g}_{\mathrm{rx}}},{{g}_{\mathrm{tx}}}}}\left(\left(n-{n^{\prime}}\right){{T}_{\mathrm{sym}}}-\tau,\left(m-m^{\prime}\right)\Delta f-\nu \right) \ne 0$ if $n=n^{\prime}$. 
Then, by substituting~\eqref{Y_TF1} into~\eqref{eq:Y_DD}, $Y^{\mathrm {DD}}[k, l]$ is given by 
\begin{equation} 
\begin{aligned}
&\hspace {-1pc}{Y^{{\mathrm{DD}}}}\left [{ {k,l} }\right] = \frac {1}{NM} \sum _{k^{\prime}=\left\lceil -\frac{N}{2} \right\rceil}^{\left\lceil \frac{N}{2} \right\rceil - 1} \sum _{l^{\prime} = 0}^{M - 1} {H^{\mathrm{DD}}_{k,l}}\left [{ {k^{\prime},l^{\prime}} }\right] \times {{X^{{\mathrm{DD}}}}\left [{k^{\prime},l^{\prime}}\right]}\\
&\qquad \qquad + {V^{{\mathrm{DD}}}}\left [{k,l} \right],
\label{Y_DD2} 
\end{aligned}
\end{equation} 
where $\mathbf{H}^{\mathrm{DD}}_{k,l}$ is channel matrix in the DD domain.
The entries of $\mathbf{H}^{\mathrm{DD}}_{k,l}$ are given by
\begin{align*}
&\hspace {-1pc}H^{\mathrm{DD}}_{k,l}[k^{\prime},l^{\prime}] \\
=&\sum_{n=0}^{N-1} \sum_{m=0}^{M-1} \iint h\left(\tau,\nu\right)
e^{-j2\pi m \Delta f \tau} e^{j 2 \pi m \frac{l-l^{\prime}}{M}}
e^{-j 2 \pi n \frac{k-k^{\prime}}{N}} \\
&e^{j2\pi \nu \left(\tau + n T_{\mathrm{sym}} + \frac{M_{\mathrm{cp}}+l^{\prime}}{M\Delta f}\right)} d\tau d\nu  \\
=&\sum_{n=0}^{N-1} \sum_{m=0}^{M-1} \sum_{i=1}^{P} \mathcal G_i \frac{|\eta_i|}{|1 + \eta_i|}  
e^{j 2\pi \eta_i \nu_i \frac{ \tau_i + \frac{M_{\mathrm{cp}} + l^{\prime}}{M \Delta f} + n T_{\mathrm{sym}}}{1+\eta_i}} \\
&e^{-j 2\pi m \Delta f\frac{\eta_i \tau_i-\frac{M_{\mathrm{cp}} + l^{\prime}}{M \Delta f}-n T_{\mathrm{sym}}}{1+\eta_i}} e^{j 2 \pi m \frac{l-l^{\prime}}{M}}
e^{-j 2 \pi n \frac{k-k^{\prime}}{N}} 
\label{eq:preH}
\tag{56}
\end{align*}
Since $|\eta_i| \gg 1$, we have $1 + \eta_i \approx \eta_i$ and $1 - \frac{1}{\eta_i} \approx 1$. Therefore,~\eqref{eq:preH} can be approximated as
\setcounter{equation}{56}
\begin{equation} 
\begin{aligned}
&\hspace {-1pc}H^{\mathrm{DD}}_{k,l}[k^{\prime},l^{\prime}] \\
{\approx}& \sum_{n=0}^{N-1} \sum_{m=0}^{M-1} \sum_{i=1}^P \mathcal G_i e^{j 2 \pi \nu_i \tau_i} e^{j 2 \pi \nu_i \left(\frac{M_{\mathrm{cp}} + l^{\prime}}{M \Delta f} + n T_{\mathrm{sym}} \right)}
e^{-j 2 \pi m \Delta f \tau_i} \\
&e^{j 2 \pi \frac{m\left(l-l^{\prime}\right)}{M}}
e^{-j 2 \pi \frac{n\left(k-k^{\prime}\right)}{N}}
e^{j2\pi mn \frac{M+M_{\mathrm{cp}}}{M\eta_i}} \\
=& \sum_{n=0}^{N-1} \sum_{m=0}^{M-1}  \sum_{i=1}^P \mathcal G_i e^{j 2 \pi \frac{k_i \left(l^{\prime} + l_i + M_{\mathrm{cp}}\right)}{N\left(M + M_{\mathrm{cp}}\right)}}
e^{j 2 \pi \frac{m\left(l - l^{\prime} - l_i\right)}{M}} \\
&e^{-j 2 \pi \frac{n\left(k - k^{\prime} - k_i\right)}{N}} 
e^{j2\pi m n \frac{M+M_{\mathrm{cp}}}{M\eta_i}}.
\end{aligned}
\end{equation} 
According to~\eqref{Y_DD2}, $Y^{\mathrm {DD}}[k, l]$ can be written as
\begin{equation} 
\begin{aligned}
&\hspace {-0.5pc}{Y^{{\mathrm{DD}}}}\left [{{k,l}}\right] \\
&= \frac {1}{NM}\sum_{i=0}^{P} \sum _{k^{\prime}=\left\lceil -\frac{N}{2} \right\rceil}^{\left\lceil \frac{N}{2} \right\rceil - 1} \sum _{l^{\prime} = 0}^{M - 1} {e^{j 2 \pi \frac{k_i \langle l - l^{\prime} \rangle_M}{N\left(M + M_{\mathrm{cp}}\right)}}}H_{k_i,l_i}^{\mathrm{DD}}\left [{ {k^{\prime},l^{\prime}} }\right] \\
&\quad \times{{X^{{\mathrm{DD}}}}\left [{ \langle k - k^{\prime} \rangle_N, ( l - l^{\prime} )_M }\right]} + {V^{{\mathrm{DD}}}}\left [{ {k,l} }\right],
\end{aligned}
\end{equation} 
where $H_{k_i,l_i}^{\mathrm {DD}}[k^{\prime},l^{\prime}]$ is given by
\allowdisplaybreaks[4]
\begin{align*}
&\hspace {-1.5pc}H_{k_i,l_i}^{\mathrm {DD}}[k^{\prime},l^{\prime}]\\
=&\sum_{m=0}^{M-1}\sum_{n=0}^{N-1}
\mathcal G_{i}e^{j2\pi\frac{k_i\left(l_i+M_{\mathrm{cp}}\right)}{N\left(M+M_{\mathrm{cp}}\right)}}e^{j2\pi\frac{n\left(k_i-k^{\prime}\right)}{N}}
e^{j2\pi\frac{m\left(l^{\prime}-l_i\right)}{M}} \\
&e^{j2\pi mn \frac{M+M_{\mathrm{cp}}}{M\eta_i}}\\
=&\mathcal G_{i}
e^{j2\pi\frac{k_i\left(l_i+M_{\mathrm{cp}}\right)}{N\left(M+M_{\mathrm{cp}}\right)}}
e^{j\pi\left(N-1\right)\left(M-1\right)\frac{M+M_{\mathrm{cp}}}{2M\eta_i}}\\
&e^{j\pi\left(N-1\right)\frac{k_i-k^{\prime}}{N}} \frac{\sin\pi N \left(\frac{k_{i}-k^{\prime}}{N}+(M-1)\frac{M+M_{\mathrm{cp}}}{2M\eta_i}\right)}{\sin\pi\left(\frac{k_{i}-k^{\prime}}{N}+(M-1)\frac{M+M_{\mathrm{cp}}}{2M\eta_i}\right)}\\
&e^{j\pi\left(M-1\right)\frac{l^{\prime}-l_{i}}{M}} \frac{\sin \pi M\left(\frac{l^{\prime}-l_{i}}{M}+n\frac{M+M_{\mathrm{cp}}}{M\eta_i}\right)}{\sin\pi\left(\frac{l^{\prime}-l_{i}}{M}+n\frac{M+M_{\mathrm{cp}}}{M\eta_i}\right)}\\
\overset{\left(b\right)}{\approx}&\mathcal G_{i}
e^{j2\pi\frac{k_i\left(l_i+M_{\mathrm{cp}}\right)}{N\left(M+M_{\mathrm{cp}}\right)}}
e^{j\pi\left(N-1\right)\left(M-1\right)\frac{M+M_{\mathrm{cp}}}{2M\eta_i}}\\
&e^{j\pi\left(N-1\right)\frac{k_i-k^{\prime}}{N}} \frac{\sin\pi N \left(\frac{k_{i}-k^{\prime}}{N}+(M-1)\frac{M+M_{\mathrm{cp}}}{2M\eta_i}\right)}{\sin\pi\left(\frac{k_{i}-k^{\prime}}{N}+(M-1)\frac{M+M_{\mathrm{cp}}}{2M\eta_i}\right)}\\
&e^{j\pi\left(M-1\right)\frac{l^{\prime}-l_{i}}{M}} \frac{\sin \pi M\left(\frac{l^{\prime}-l_{i}}{M}+(N-1)\frac{M+M_{\mathrm{cp}}}{2M\eta_i}\right)}{\sin\pi\left(\frac{l^{\prime}-l_{i}}{M}+(N-1)\frac{M+M_{\mathrm{cp}}}{2M\eta_i}\right)}.
\tag{59}
\end{align*}
Note that we substitute $n$ with $\frac{N-1}{2}$ in $(b)$ to yield a closed-form expression. The proof is complete.

 
%

\bibliographystyle{IEEEtran}
\bibliography{paper_reference}

\begin{IEEEbiography}[{\includegraphics[width=1in,height=1.25in,clip,keepaspectratio]
{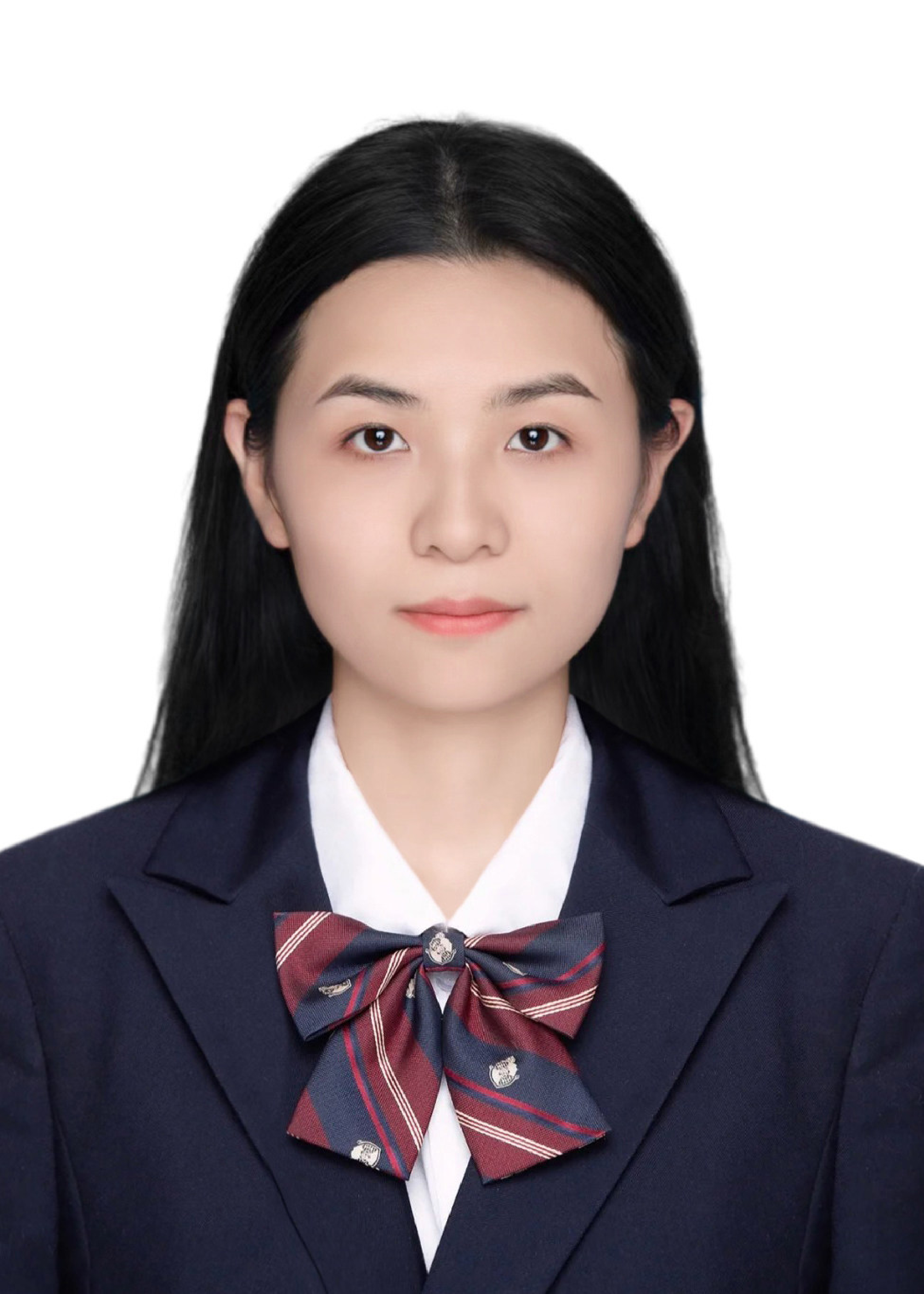}}]{Yiyan Cheng} received the B.S. degree in Communication Engineering from Sichuan Normal University, China, in 2023, and the M.S. degree in Information and Communication Engineering from Beijing University of Posts and Telecommunications, China, in 2026. Her research interests include satellite communications and orthogonal time frequency space (OTFS) modulation.
\end{IEEEbiography}

\begin{IEEEbiography}[{\includegraphics[width=1in,height=1.25in,clip,keepaspectratio]{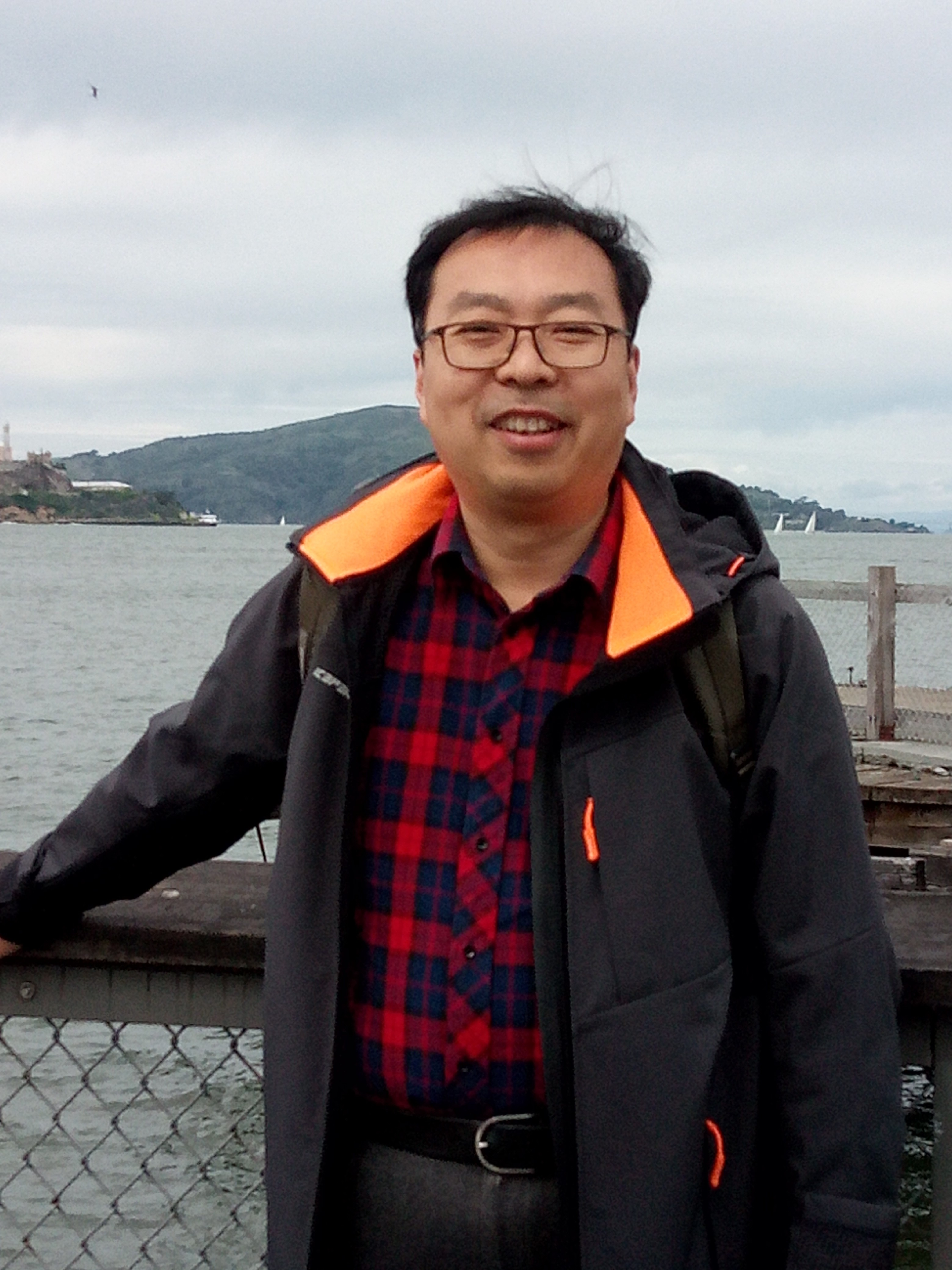}}]{Tiejun Lv}(Senior Member, IEEE) received the M.S. and Ph.D. degrees in electronic engineering from the University of Electronic Science and Technology of China (UESTC), Chengdu, China, in 1997 and 2000, respectively. From January 2001 to January 2003, he was a Postdoctoral Fellow at Tsinghua University, Beijing, China. In 2005, he was promoted to Full Professor at the School of Information and Communication Engineering, Beijing University of Posts and Telecommunications (BUPT). From September 2008 to March 2009, he was a Visiting Professor with the Department of Electrical Engineering at Stanford University, Stanford, CA, USA. He is the author of four books, one book chapter, more than 160 published journal papers and 220 conference papers on the physical layer of wireless mobile communications. His current research interests include signal processing, communications theory and networking. He was the recipient of the Program for New Century Excellent Talents in University Award from the Ministry of Education, China, in 2006. He received the Nature Science Award from the Ministry of Education of China for the hierarchical cooperative communication theory and technologies in 2015 and the Shaanxi Higher Education Institutions Outstanding Scientific Research Achievement Award in 2025.
\end{IEEEbiography}

\begin{IEEEbiography}[{\includegraphics[width=1in,height=1.25in,clip,keepaspectratio]
{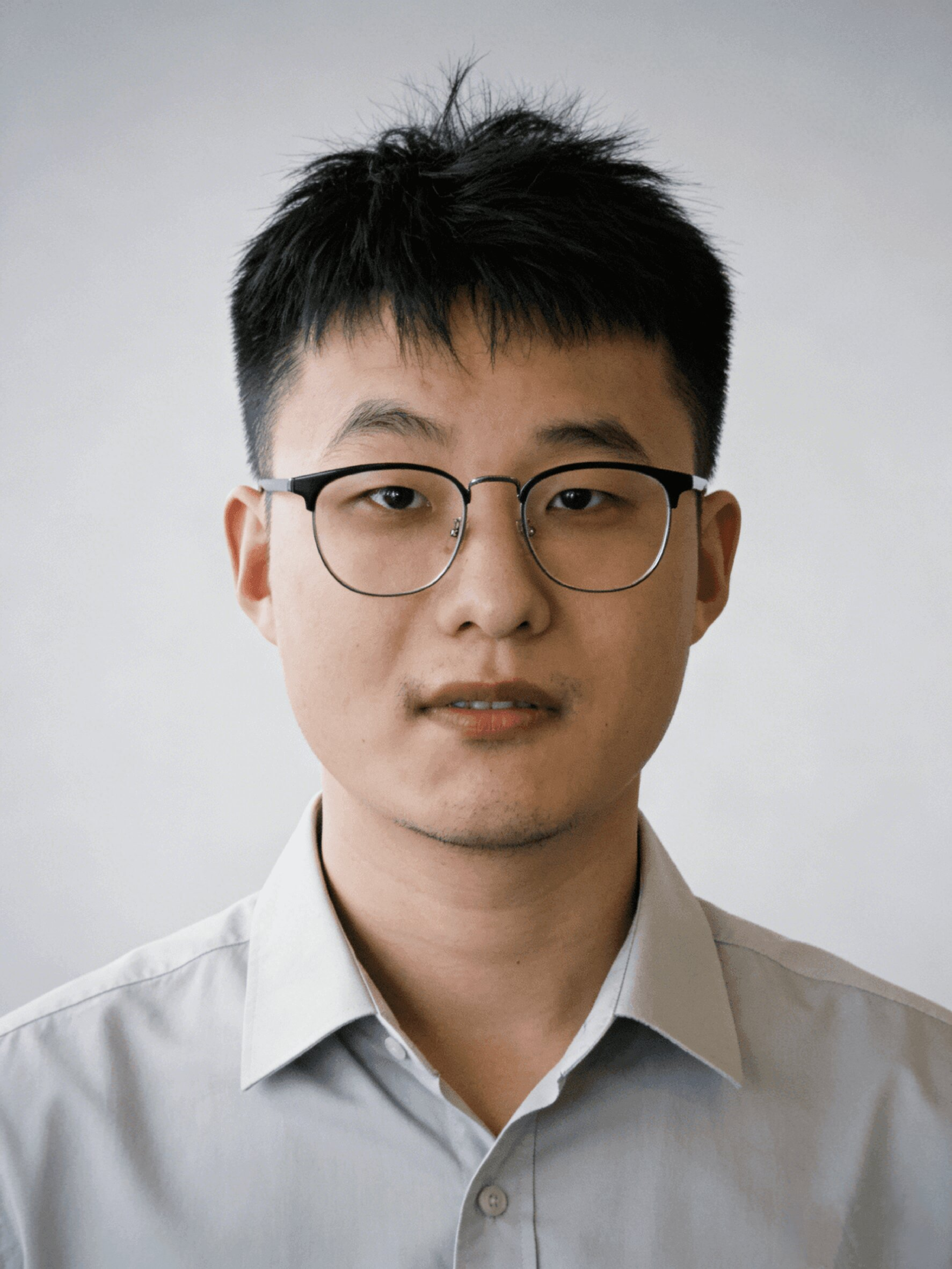}}]{Yashuai Cao} received the B.E. and Ph.D. degrees in communication engineering from Chongqing University of Posts and Telecommunications (CQUPT) and Beijing University of Posts and Telecommunications (BUPT), China, in 2017 and 2022, respectively. From 2022 to 2023, he was a lecturer in the Department of Electronics and Communication Engineering, North China Electric Power University (NCEPU), Baoding. From 2023 to 2025, he was a Postdoctoral Research Fellow with the Department of Electronic Engineering, Tsinghua University, Beijing, China. He is currently a Distinguished Associate Professor with the School of Artificial Intelligence, University of Science and Technology Beijing (USTB), Beijing, China. His research interests include Stacked Intelligent Metasurface, Environment-Aware Communications, and Channel Knowledge Map.
\end{IEEEbiography}

\begin{IEEEbiography}[{\includegraphics[width=1in,height=1.25in,clip,keepaspectratio]
{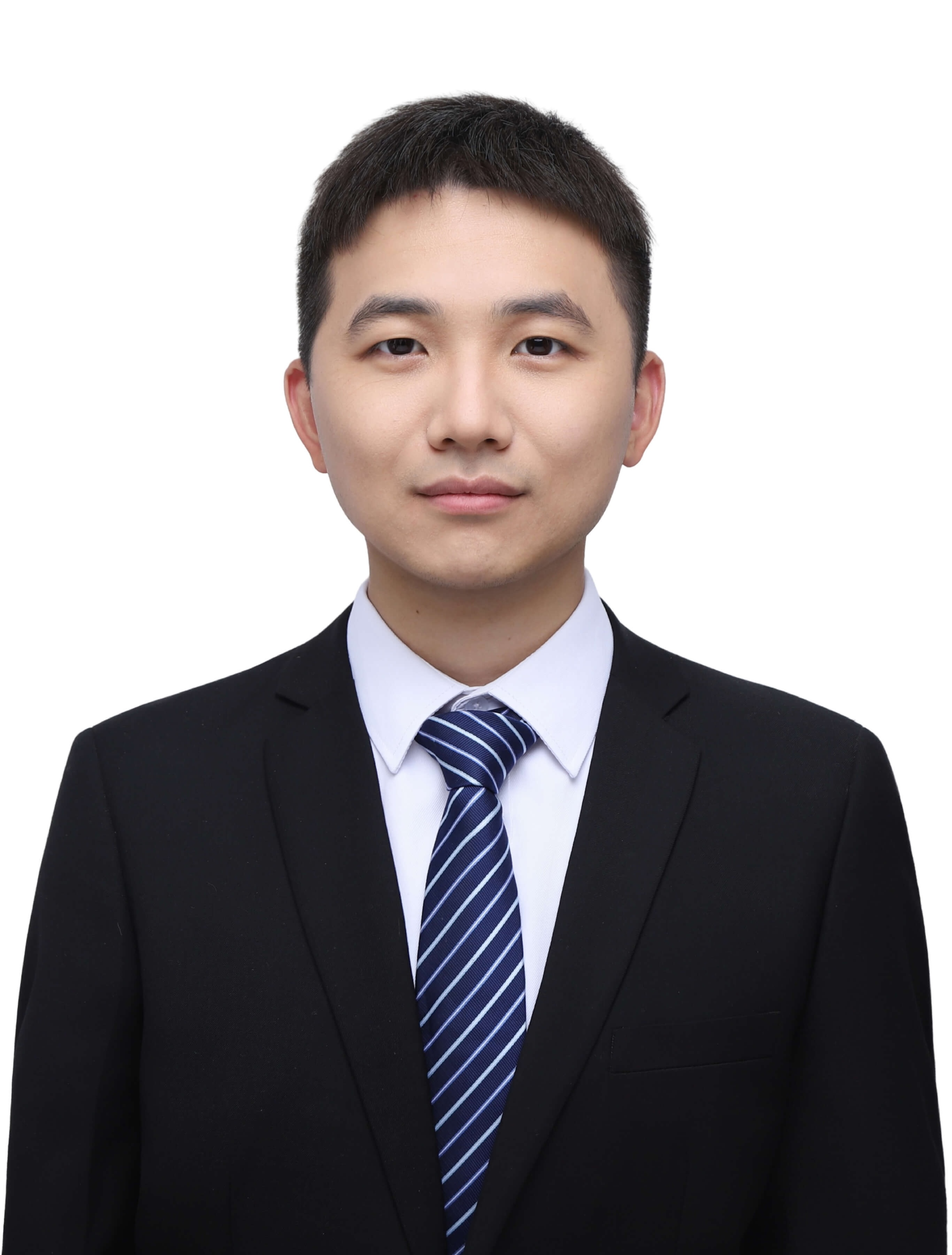}}]{Xuehan Wang} (Member, IEEE) received the B.Eng. and Ph.D. degrees both from the Department of Electronic Engineering, Tsinghua University, Beijing, China, in 2021 and 2026, respectively. He is currently with the China Mobile Research Institute. His research interests include wireless networks, delay-Doppler domain signal processing, mobile communication and underwater acoustic signal processing.
\end{IEEEbiography}

\begin{IEEEbiography}[{\includegraphics[width=1in,height=1.25in,clip,keepaspectratio]{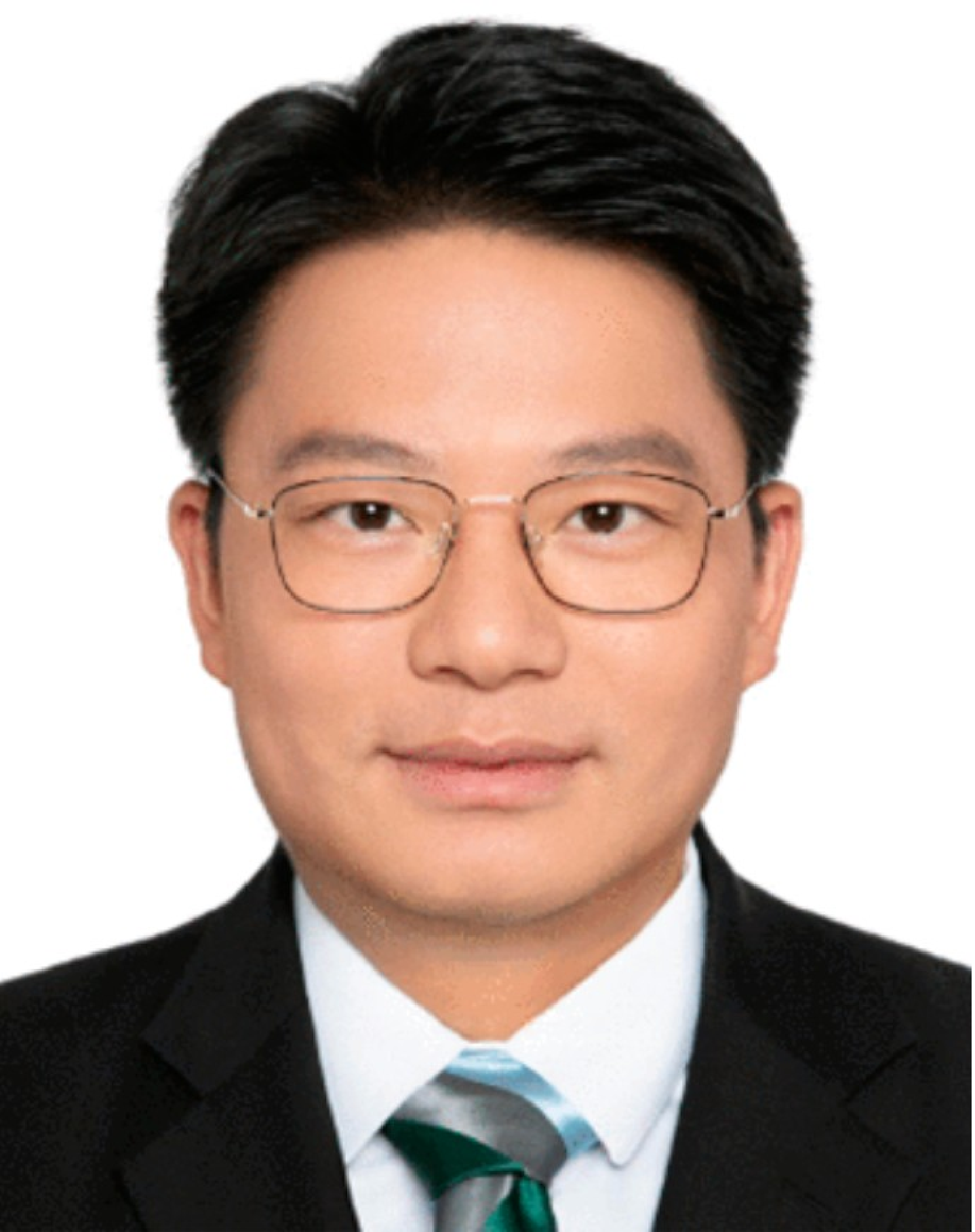}}]{Mugen Peng} received the Ph.D. degree in communication and information systems from the Beijing University of Posts and Telecommunications (BUPT), Beijing, China, in 2005. Afterward, he joined BUPT, where he has been the Dean of the School of Information and Communication Engineering since June 2020 and the Deputy Director of the State Key Laboratory of Networking and Switching Technology since October 2018. In 2014, he was also an Academic Visiting Fellow with Princeton University, USA. He has authored or coauthored over 150 refereed IEEE journal articles and over 250 conference proceeding papers. His main research areas include wireless communication theory, radio signal processing, cooperative communication, cloud communication, and the Internet of Things. He was a recipient of the 2018 Heinrich Hertz Prize Paper Award, the 2014 IEEE ComSoc AP Outstanding Young Researcher Award, and the Best Paper Award in the ICC 2022, ICCC 2020, IEEE WCNC 2015, and JCN 2016. He is currently on the Editorial/Associate Editorial Board Member of the IEEE Network, the IEEE Communications Magazine, the IEEE Internet of Things Journal, the IEEE Transactions on Vehicular Technology, the IEEE Transactions on Network Science and Engineering, the Intelligent and Converged Networks, and the Digital Communications and Networks (DCN).
\end{IEEEbiography}

\end{document}